\documentclass[11pt,letterpaper]{article}
\usepackage{fullpage}
\usepackage[numbers]{natbib}

\usepackage{booktabs} 

\usepackage{amsmath,amsthm,amsfonts,amssymb}
\usepackage{bm}
\usepackage{bbm}
\usepackage{xspace}
\usepackage{color}
\usepackage[hidelinks,colorlinks,allcolors=black,citecolor=magenta]{hyperref}
\let\oldciteauthor=\citeauthor
\def\citeauthor#1{\hypersetup{citecolor=black}\oldciteauthor{#1}}
\usepackage{wrapfig}
\usepackage{anyfontsize}
\usepackage{subcaption}
\usepackage{cleveref}
\usepackage{tikz,pgfplots}
\usepackage{algpseudocode,algorithm}
\pgfplotsset{compat=1.11}

\usepackage{tikz}
\usetikzlibrary{arrows,matrix,positioning,fit}
\usepackage{nicefrac}

\usepackage{pifont}
\newcommand{\cmark}{\ding{51}\xspace}%
\newcommand{\xmark}{\ding{55}\xspace}%

\newcommand{\RPS}{\textup{\textsc{RPS}}}
\newcommand{\Eating}{\textup{\textsc{Eating}}}

\newtheorem{proposition}{Proposition}
\newtheorem{corollary}{Corollary}
\newtheorem{theorem}{Theorem}
\newtheorem{lemma}{Lemma}

\theoremstyle{definition}
\newtheorem{definition}{Definition}
\newtheorem{examplex}{Example}
\newenvironment{example}{\pushQED{\qed}\examplex}{\popQED\endexamplex}

\DeclareMathOperator*{\argmax}{\arg\max}

\newcommand{\set}[1]{\{#1\}}
\newcommand{\floor}[1]{\lfloor{#1}\rfloor}
\newcommand{\ceil}[1]{\lceil{#1}\rceil}

\renewcommand{\hat}{\widehat}
\renewcommand{\H}{\mathcal{H}}

\renewcommand{\succeq}{\succcurlyeq}

\usepackage{mathtools}

\begin{document}

\title{Best of Both Worlds:\\Ex-Ante and Ex-Post Fairness in Resource Allocation}

\author{
	Rupert Freeman\\Microsoft Research New York\\\texttt{rupert.freeman@microsoft.com}
	\and
	Nisarg Shah\\University of Toronto\\\texttt{nisarg@cs.toronto.edu}
	\and
	Rohit Vaish\\Rensselaer Polytechnic Institute\\\texttt{vaishr2@rpi.edu}
}
\date{}
\maketitle
\begin{abstract}
	We study the problem of allocating indivisible goods among agents with additive valuations. When randomization is allowed, it is possible to achieve compelling notions of fairness such as envy-freeness, which states that no agent should prefer any other agent's allocation to her own. When allocations must be deterministic, achieving exact fairness is impossible but approximate notions such as envy-freeness up to one good can be guaranteed. Our goal in this work is to achieve both simultaneously, by constructing a randomized allocation that is exactly fair ex-ante and approximately fair ex-post. The key question we address is whether ex-ante envy-freeness can be achieved in combination with ex-post envy-freeness up to one good. We settle this positively by designing an efficient algorithm that achieves both properties simultaneously. If we additionally require economic efficiency, we obtain an impossibility result. However, we show that economic efficiency and ex-ante envy-freeness can be simultaneously achieved if we slightly relax our ex-post fairness guarantee. On our way, we characterize the well-known Maximum Nash Welfare allocation rule in terms of a recently introduced fairness guarantee that applies to groups of agents, not just individuals.
\end{abstract}

\section{Introduction}
\label{sec:Introduction}

\emph{What is fair? What is just?} These questions have captivated many brilliant minds in human history, from the ancient Greek philosopher Plato to the renowned political philosopher John Rawls~\cite{Rawls71}. In the last decade or so, with algorithms increasingly making decisions that affect human lives, fairness has been a subject of intense research in the context of algorithmic decision-making~\cite{CPFG+17}. Some of this research draws on the --- almost a century old~\cite{Stein48} --- literature on fair resource allocation. 

The central problem in fair resource allocation is, as the name suggests, to fairly allocate a set of resources (often dubbed \emph{goods}) among a set of individuals (often dubbed \emph{agents}) who have different preferences over the resources. Much of the early economic work treats the resources as \emph{divisible}. For example, in the classic cake-cutting setting~\cite{Stein48}, the ``cake'' is to be split among the agents. An agent's preferences are expressed through a valuation function, which places a value on every subset of the cake. Traditionally, this valuation function is assumed to be \emph{additive}, that is, an agent's value for receiving disjoint parts of the cake is the sum of her values for the individual parts. The possibility of finely dividing the cake allows allocating it with compelling fairness guarantees. For example, \citet{Strom80} showed that one can always allocate the cake in a way that is \emph{envy-free} (EF), i.e., such that no agent values the share of the cake allocated to another agent more than the share allocated to herself.  

Now imagine if the goods are \emph{indivisible}, i.e., each good must be entirely allocated to a single agent. In other words, the allocation to each agent is a subset of the set of goods. If we are allowed to randomize,
then we can choose an agent uniformly at random and allocate the entire set of goods to her. This is trivially \emph{ex-ante envy-free}, since all agents receive the same distribution over bundles of goods. However, this allocation induces a large amount of envy \emph{ex-post}, since one agent receives everything and all others receive nothing.

It is obvious that some amount of ex-post envy is unavoidable; imagine two agents liking a single good, which must be given to one of them, leaving the other envious. This difficulty has led to significant research on fairness in deterministic allocations of indivisible goods in economics and computer science~\cite{BCMM16}, a bulk of which focuses on relaxed fairness properties which can be guaranteed on all instances. One compelling example is \emph{envy-freeness up to one good} (EF1), which requires that the envy of any agent toward another agent can be removed by the elimination of at most one good from the envied agent's bundle.

It is known that a deterministic EF1 allocation always exists~\cite{LMMS04,Bud11}. Such an allocation keeps the ex-post envy limited, but it may still be perceived as unfair if some agents are systematically favored over others. 
For example, if two agents agree on their ordinal preference ranking over two goods but have different strengths of preference, the popular Maximum Nash Welfare (MNW)~\citep{Nash50b,KN79,CKMP+19} allocation rule\footnote{The MNW allocation is the one that maximizes the \emph{product} of agent utilities. This is the rule employed by the popular fair division website Spliddit (\url{http://spliddit.org}).} assigns the more valuable good to the agent with higher intensity of preference for it over the less valuable good, and the less valuable good to the other agent. If this allocation happens repeatedly, the scale will always be tipped towards the same agent and their advantage over time will become large.

One might hope that, at least for the case of repeated allocation, more dedicated solutions could be designed. For instance, the goods could be allocated by the round-robin method --- where the agents choose goods one at a time in a fixed order --- in each round, while cycling through all $n!$ possible orderings of the agents so that no one is repeatedly favored by the mechanism. However, as we will see later, this seemingly-fair solution is not enough to guarantee envy-freeness in the long run. Further, it does nothing to alleviate unfairness in a one-shot setting, where ex-ante envy-freeness may still be desired as agents may be unwilling to participate in a system in which they know they will certainly be treated less favorably than their competitors.

This motivates a natural question. \emph{Can we retain envy-freeness up to one good as an ex-post guarantee, and simultaneously obtain (exact) envy-freeness ex-ante?} In other words, can we always randomize over EF1 allocations such that the resulting randomized allocation is EF? We show that the answer to this natural and elegant question is \emph{yes}. More generally, we study various combinations of ex-ante and ex-post fairness and efficiency guarantees, and identify combinations that can (and cannot) be achieved simultaneously. Our constructive results yield efficient algorithms; these improve upon prior algorithms which provide either only ex-ante or only ex-post guarantees, thus paving the way for fairer resource allocation in practice.

\subsection{Our Results}\label{sec:results}

Our main result is the development of a novel algorithm, \emph{Recursive Probabilistic Serial}, that produces an ex-ante envy-free (EF) distribution over deterministic allocations that each satisfy envy-freeness up to one good (EF1). Our algorithm is an adaptation of the classic \emph{Probabilistic Serial (PS)} algorithm of~\citet{BM01} that inherits much of the desirable behavior of PS (in particular, ex-ante envy-freeness), while also allowing a simple and natural decomposition over EF1 allocations. While a na\"{\i}ve version of our algorithm yields a distribution over a possibly-exponential number of deterministic allocations, we show that a support size polynomial in the number of agents and goods is sufficient, thus yielding an efficient variant.

In addition to ex-ante EF and ex-post EF1, one may want to achieve the economic efficiency notion of Pareto optimality, which states that it should be impossible to find an allocation that improves some agent's utility without reducing any other agent's. In Section~\ref{sec:Impossibility} we show that it is impossible to achieve ex-ante Pareto optimality (that is, with respect to the randomized allocation) in conjunction with ex-ante EF and ex-post EF1. However, in Section~\ref{sec:Prop1_EF11} we show that strong ex-ante guarantees --- in terms of both fairness and economic efficiency --- can be achieved if we are willing to compromise on the ex-post guarantee. In particular, we are able to achieve ex-ante \emph{group fairness} (GF)~\citep{CFSV19}, which generalizes both envy-freeness and Pareto optimality, in conjunction with two ex-post fairness properties that are incomparable but are both implied by EF1: \emph{proportionality up to one good (Prop1)}~\citep{CFS17} and \emph{envy-freeness up to one good more-and-less (EF}$_1^1$\emph{)}~\citep{BK19}. Our algorithm uses a rounding of the well-known MNW allocation, and in Section~\ref{sec:characterization}, we provide a novel characterization which shows that this is the only allocation rule that can be used to achieve the desired properties. Table~\ref{tab:results} summarizes our results.

Finally, in \Cref{sec:bads}, we extend our results to the case of bads. In particular, for Recursive Probabilistic Serial, we show that the natural extension only manages to achieve ex-post EF2 (in addition to ex-ante EF), but ex-post EF1 can be recovered through a simple modification. 

\begin{table}
	\centering
	\begin{tabular}{c| c | c | c| c}
	  	& Prop1 + EF$_1^1$ & EF1 & EF1 + PO & EF1 + fPO\\ \hline
	  Prop & \cmark & \cmark & \textbf{?} & \xmark (Thm~\ref{thm:impossible-prop-ef1-fpo}) \\
	  EF & \cmark & \cmark (Thm~\ref{thm:eating_ef}) & \bf{?} & \xmark \\
	  GF & \cmark (Cor~\ref{fgf-prop1}) & \xmark (Thm~\ref{thm:impossible-prop-ef1-fpo}) & \xmark & \xmark \\
	\end{tabular}
	\caption{Summary of our results. Each row corresponds to an ex-ante guarantee while each column corresponds to an ex-post guarantee. A \cmark indicates possibility (and polynomial-time computation), a \xmark indicates impossibility, and a \textbf{?} indicates an open problem. Note that positive results propagate above and to the left, while negative results propagate below and to the right.}
	\label{tab:results}
\end{table}

\subsection{Related Work}\label{sec:related}

A large body of work in computer science and economics has focused on finding exactly ex-ante fair randomized allocations, as well as approximately fair deterministic allocations, and we cite those works as appropriate throughout the paper. Combining the two approaches was recently listed as an ``interesting challenge'' by ~\citet{aziz2019probabilistic}, however, little work has focused on this problem. 
Two exceptions are~\citet{aleksandrov2015online} and~\citet{BCKM13}. \citet{aleksandrov2015online} consider randomized allocation mechanisms for an online fair division problem and analyze their ex-ante and ex-post fairness guarantees. The style of their results is very similar to ours, however they restrict attention to binary utilities, which simplifies the problem significantly. \citet{BCKM13} study the problem of implementing a general class of random allocation mechanisms subject to ex-post constraints, although the ex-post constraints are not the same as ours. In particular, they do not consider ex-post axiomatic guarantees from the fair division literature as we do.

In the random assignment literature in economics, the idea of constructing a fractional assignment and implementing it as a lottery over pure assignments was introduced by~\citet{HZ79}. Later work has studied both ex-ante and ex-post fairness and efficiency guarantees provided by mechanisms in this setting~\citep{BM01,abdulkadirouglu1998random,chen2002improving,nesterov2017fairness}, but most of this work studies ordinal utilities and does not consider approximate notions of ex-post fairness.\footnote{The standard random assignment setting has $n$ agents, $n$ items, and requires that each agent receive exactly one item. The notion of ex-post fairness that we use in this work, envy-freeness up to one good, is vacuous in this restricted setting.} \citet{gajdos2002fairness} study the relationship between ex-ante and ex-post fairness but in their model the randomness comes from nature, not the allocation rule. Other work studies the problem of implementing a fractional outcome over deterministic outcomes subject to (possibly soft) constraints~\citep{BCKM13,akbarpour2019approximate}, but the constraints allowed by these papers do not fully capture our ex-post fairness notions.

In Section~\ref{sec:characterization}, we obtain a characterization of the maximum Nash welfare (MNW) rule --- this is equivalent to competitive equilibrium with equal incomes (CEEI)~\cite{EG59,AI82}, and is also known as the competitive rule or the Walras rule --- for fractional allocations using group fairness and a property called replication-invariance. There are several related characterizations in the literature. Assuming monotonic and strictly convex preferences, it is known that CEEI is the only replication-invariant allocation rule that is in the core~\citep{DS63}, or group envy-free~\cite{Var74}, or proportional and Pareto optimal~\cite{Tho88}. While the core, group envy-freeness, or proportionality together with Pareto optimality are all weaker than group fairness (making these characterizations seemingly stronger), to the best of our knowledge, these characterizations only apply for \emph{strictly convex} preferences, and not to additive preferences. We also note that our proof technique is significantly different from the proofs of these prior characterizations because additive preferences do not admit some of the nice structural implications that strictly convex preferences do when combined with replication invariance, as we note in Section~\ref{sec:characterization}.

\section{Preliminaries}
\label{sec:Preliminaries}

For any positive integer $r \in \mathbb{N}$, define $[r] \coloneqq \{1,\dots,r\}$. Let $N = [n]$ denote a set of \emph{agents}, and $M$ denote a set of \emph{goods} where $m \coloneqq |M|$. 

\paragraph{Fractional and Randomized Allocations}
A \emph{fractional allocation} of the items in $M$ to the agents in $N$ is specified by a non-negative $n \times m$ matrix $A \in [0,1]^{n \times m}$ such that for every item $j \in M$, we have $\sum_{i \in N} A_{i,j} \leq 1$; here, $A_{i,j}$ denotes the fraction of item $j$ assigned to agent $i$. We say that a fractional allocation $A$ is \emph{complete} if $\sum_{i \in N} A_{i,j} = 1$ for every $j \in M$, and call it \emph{partial} otherwise. 

A fractional allocation $A$ is \emph{integral} if $A_{i,j} \in \{0,1\}$ for every $i \in N$ and $j \in M$. For integral allocations, we will find it convenient to denote the binary vector $A_i = (A_{i,j})_{j \in M}$ as a set $A_i \coloneqq \{j \in M : A_{i,j} = 1\}$. We will refer to $A_i$ as the \emph{bundle} of items assigned to agent $i$, and denote the allocation $A$ as an ordered tuple of bundles $A = (A_1,\dots,A_n)$. When we simply say `an allocation', it will mean a fractional allocation, unless otherwise clear from the context. For notational clarity, we use letters $X$ or $Y$ for fractional allocations, and write $A$ or $B$ for integral allocations. Let $\mathcal{X}$ be the set of all (complete or partial) fractional allocations.

A \emph{randomized allocation} is a lottery over integral allocations (we denote them by bold letters for clarity). Formally, a randomized allocation $\mathbf{X}$ is specified by a set of $\ell \in \mathbb{N}$ ordered pairs $\{(p^k,A^k)\}_{k \in [\ell]}$, where, for every $k \in [\ell]$, $A^k$ is an integral allocation implemented with probability $p^k \in [0,1]$, and $\sum_{k \in [\ell]} p^k = 1$. The \emph{support} of $\mathbf{X}$ is the set of integral allocations $\set{A^1,\dots,A^\ell}$. 

A randomized allocation $\mathbf{X} \coloneqq \{(p^k,A^k)\}_{k \in [\ell]}$ is naturally associated with the fractional allocation $X \coloneqq \sum_{k \in [\ell]} p^k A^k$, where $X_{i,j}$ is the (marginal) probability of agent $i$ receiving good $j$ under $\mathbf{X}$. In this case, we say that randomized allocation $\mathbf{X}$ \emph{implements} fractional allocation $X$. There may be many randomized allocations implementing a given fractional allocation. In \Cref{sec:EF_EF1,sec:Prop1_EF11}, we discuss and make use of two known techniques for implementing a given fractional allocation.

\paragraph{Preferences}
Each agent $i \in N$ has an \emph{additive} valuation function $v_i$, where $v_{i,j} \ge 0$ denotes the agent's utility for fully receiving good $j \in M$. Later, in Section~\ref{sec:bads}, we consider allocation of \emph{bads}, where $v_{i,j} \le 0$ for each agent $i \in N$ and bad $j \in M$. Note that $v_i$ induces a weak order $\succeq_i$ over goods where $g_j \succeq_i g_k$ if and only if $v_{i,j} \ge v_{i,k}$, and $g_j \succ_i g_k$ if and only if $v_{i,j} > v_{i,k}$. The utility of agent $i$ under an allocation $X \in \mathcal{X}$ is given by $v_i(X_i) = \sum_{j \in M} X_{i,j}\cdot v_{i,j}$. We assume that for each good $j \in M$, there exists at least one agent $i \in N$ with $v_{i,j} > 0$. This is without loss of generality as goods valued zero by everyone can be allocated arbitrarily.

\paragraph{Allocation rule}
A fair division instance $I$ is defined by the triple $(N,M, (v_i)_{i \in N})$. We let $\mathcal{I}$ denote the set of all instances. An \emph{allocation rule} $f: \mathcal{I} \to 2^\mathcal{X}$ maps instances to (sets of) allocations.

\begin{definition}[Fractional Maximum Nash Welfare Rule]
Given an instance $I \in \mathcal{I}$, the fractional Maximum Nash Welfare (MNW) rule returns all fractional allocations that maximize the product of agents' utilities, i.e., $\operatorname{MNW}(I) \coloneqq \argmax_{X \in \mathcal{X}}\ \Pi_{i \in N}\ v_i(X_i)$. We refer to an allocation $A \in \operatorname{MNW}(I)$ as a fractional MNW allocation. Integral MNW allocations maximize the product of agents' utilities across all integral allocations; however, there is a subtle tie-breaking involved in that case~\cite{CKMP+19}.
\end{definition}

We now discuss a number of properties concerning fairness and efficiency of allocations and allocation rules. 

\subsection{Fairness and Efficiency Properties}

For any property $\langle P \rangle$ defined for a fractional allocation, we say that a randomized allocation $\mathbf{X}$ satisfies $\langle P \rangle$ \emph{ex-ante} if the fractional allocation $X$ it implements satisfies $\langle P \rangle$. Similarly, for any property $\langle Q \rangle$ defined for an integral allocation, we say that a randomized allocation $\mathbf{X}$ satisfies $\langle Q \rangle$ \emph{ex-post} if every integral allocation in its support satisfies $\langle Q \rangle$.

We begin with a classic fairness notion called \emph{proportionality}.

\begin{definition}[Proportionality (Prop)~\cite{Stein48}]
	An allocation $X$ is \emph{proportional} if for each agent $i \in N$, $v_i(X_i) \ge v_i(\vec{1}^m)/n$, where $v_i(\vec{1}^m)$ is agent $i$'s utility for receiving all goods fully.
\end{definition}

For complete allocations, the following fairness guarantee is logically stronger than proportionality, and has been a subject of substantial research. 

\begin{definition}[Envy-Freeness (EF)~\cite{Fol67}]
	An allocation $X$ is \emph{envy-free} if for every pair of agents $i,h \in N$, we have $v_i(X_i) \ge v_i(X_h)$.
\end{definition}

Given allocations $X$ and $Y$, we say that agent $i$ \emph{SD-prefers} $X_i$ to $Y_i$, written $X_i \succeq_i^{SD} Y_i$, if for every good $g \in M$, we have that $\sum_{g_j \in \set{g' \in M : g' \succeq_i g} } X_{i,j} \ge \sum_{g_j \in \set{g' \in M : g' \succeq_i g} } Y_{i,j}$. Here, ``SD'' refers to first-order stochastic dominance. It is easy to check that $X_i \succeq_i^{SD} Y_i$ is equivalent to $v'_i(X_i) \ge v'_i(Y_i)$ under every additive valuation $v'_i$ consistent with the ordinal preference relation $\succeq_i$. 

\begin{definition} [SD-Envy-Freeness (SD-EF)~\cite{BM01}]
	An allocation $X$ is \emph{SD-envy-free} if for every pair of agents  $i,h \in N$, we have $X_i \succeq_i^{SD} X_h$.
\end{definition}

Equivalently, an allocation is SD-envy-free if it is envy-free under any additive valuation functions of the agents consistent with $(\succeq_i)_{i \in N}$. Next, we discuss economic efficiency of allocations.

\begin{definition}[Fractional Pareto Optimality (fPO)~\citep{BKV18} and Pareto Optimality (PO)]
An allocation $X$ is \emph{fractionally Pareto optimal} if there is no fractional allocation $Y$ that Pareto-dominates it, i.e., satisfies $v_i(Y_i) \ge v_i(X_i)$ for all agents $i \in N$ and at least one inequality is strict. An integral allocation $A$ is \emph{Pareto optimal} if there is no integral allocation $B$ that Pareto-dominates it. 
\end{definition}

For fractional allocations, fPO is traditionally just referred to as Pareto optimality, so we will use fPO and PO interchangeably. However, for integral allocations, fPO is stronger than PO. 

\begin{proposition}\label{prop:po-fpo}
	If a randomized allocation is ex-ante PO, then it is also ex-post fPO.
\end{proposition}
\begin{proof}
	If a randomized allocation $\mathbf{X} \coloneqq \{(p^k,A^k)\}_{k \in [\ell]}$ implementing a fractional allocation $X$ is not ex-post fPO, then for some $k \in [\ell]$, integral allocation $A^k$ must be Pareto dominated by a fractional allocation, say $Y$. Then, the fractional allocation $X' \coloneqq p^k \cdot Y + \sum_{r \neq k} p^r \cdot A^r$ Pareto-dominates $X$, which implies that $\mathbf{X}$ is not ex-ante PO.
\end{proof}

The following property generalizes various properties mentioned above. Given any set $S \subseteq N$ of agents, we write $\cup_{i \in S} X_i$ to denote the union of the fractional allocations to agents in $S$, i.e., $\cup_{i \in S} X_i \coloneqq (\sum_{i \in S} X_{i,j})_{j \in M}$.

\begin{definition}[Group Fairness (GF)~\cite{CFSV19}]\label{defn:GF}
	An allocation $X$ is \emph{group fair} if for all non-empty subsets of agents $S,T \subseteq N$, there is no fractional allocation $Y$ of $\cup_{i \in T} X_i$ to the agents in $S$ such that $\frac{|S|}{|T|} \cdot v_i(Y_i) \ge v_i(X_i)$ for all agents $i \in S$ and at least one inequality is strict.
\end{definition}

Note that imposing the above constraint over restricted $(S,T)$ pairs can recover properties such as proportionality ($|S|=1, T=N$), envy-freeness ($|S|=|T|=1$), and fractional Pareto optimality ($S=T=N$). 

While the fairness properties mentioned above are compelling, they cannot be guaranteed with integral allocations. For such allocations, it is meaningful to consider approximate fairness notions, some of which are discussed below.

\begin{definition}[Proportionality Up To One Good (Prop1)~\cite{CFS17}]
	An integral allocation $A$ is \emph{proportional up to one good} if for every agent $i \in N$, either $v_i(A_i) \geq v_i(\vec{1}^m)/n$ or there exists a good $j \notin A_i$ such that $v_i(A_i) + v_{i,j} \geq v_i(\vec{1}^m)/n$, where $v_i(\vec{1}^m)$ is agent $i$'s utility for receiving all goods fully.
\end{definition}

\begin{definition}[Envy-Freeness Up To One Good (EF1)~\cite{LMMS04,Bud11}]
	An integral allocation $A$ is \emph{envy-free up to one good} if for every pair of agents $i,h \in N$ such that $A_h \neq \emptyset$, we have $v_i(A_i) \geq v_i(A_h \setminus \{j\})$ for some good $j \in A_h$.
\end{definition}

All of our possibility and impossibility results also apply for the following natural strengthening of EF1, which has been formally defined in subsequent work~\cite{Aziz20}.

\begin{definition}[SD-Envy-Freeness up to One Good (SD-EF1)]
	An integral allocation $A$ is \emph{SD-envy-free up to one good} if for every pair of agents $i,h \in N$ such that $A_h \neq \emptyset$, we have $A_i \succeq_i^{SD} A_h \setminus \set{j}$ for some good $j \in A_h$.
\end{definition}

This is equivalent to agent $i$ not envying agent $h$ up to one good under every additive valuation consistent with the ordinal preference relation $\succeq_i$. 

The following property is a relaxation of EF1 and enjoys strong algorithmic support in conjunction with PO, for goods (Section~\ref{sec:Prop1_EF11}) as well as bads (Section~\ref{sec:bads})~\cite{BK19,BS19}.

\begin{definition}[Envy-Freeness Up To One Good More-and-Less (EF$_1^1$)~\cite{BK19}]
	An integral allocation $A$ is \emph{envy-free up to one good more-and-less} if for every pair of agents $i,h \in N$ such that $A_h \neq \emptyset$, we have $v_i(A_i \cup \{j_i\}) \geq v_i(A_h \setminus \{j_h\})$ for some goods $j_i \notin A_i$ and $j_h \in A_h$.
\end{definition}

These three guarantees are logically related as follows: 
SD-EF1$\implies$EF1$\implies$EF$_1^1$ and SD-EF1$\implies$EF1$\implies$Prop1.

\section{Possibility: Ex-ante EF + Ex-post EF1}
\label{sec:EF_EF1}

This section describes our main result that ex-ante envy-freeness can be achieved in conjunction with ex-post envy-freeness up to one good (EF1).

A natural approach towards this question is to start with the round-robin method. Under this method, we fix an agent ordering, and then agents take turns picking one good at a time in a cyclic fashion. At each step, an agent picks her most valuable good that is still available. It is well-known that for any agent ordering, this method produces an integral allocation that is EF1~\cite{CKMP+19}. Further, it is easy to see that uniformly randomizing the agent ordering --- the so-called \emph{randomized round-robin method} --- also achieves ex-ante proportionality.

\begin{proposition}\label{prop:RRR-Prop}
	With additive valuations, randomized round-robin is ex-ante proportional and ex-post envy-free up to one good.
\end{proposition}
\begin{proof}
	We already noticed that randomized round-robin is ex-post EF1. For ex-ante Prop, fix an agent $i$, and suppose, without loss of generality, that $v_{i,1} \ge v_{i,2} \ge \ldots \ge v_{i,m}$. For ease of presentation, suppose that $m=cn$ for some positive integer $c$. When agent $i$ is $k$-th in the ordering, her value for her allocation is at least $v_{i,k} + v_{i,k+n}+\ldots+v_{i,k+(c-1)n}=\sum_{c'=1}^c v_{i,(c'-1)n+k}$ because the $c'$-th good she chooses must be no less preferred than her $(c'-1)n+k$-th most preferred good. Since she appears in every position in the ordering with probability $1/n$, her expected value is at least $\frac{1}{n} \sum_{k=1}^m \sum_{c'=1}^c v_{i,(c'-1)n+k} = \frac{1}{n} \sum_{j=1}^m g_{i,j}$, which matches her proportionality guarantee.
\end{proof}

However, as we observe below, this approach fails to achieve ex-ante envy-freeness.

\begin{example}
Consider an instance with three agents $a_1,a_2,a_3$ and three goods $g_1,g_2,g_3$ as shown below (where $0 < \varepsilon < 0.25$):\footnote{We note that a similar instance was used by \citet[Proposition 13]{aziz2014cake} in the divisible goods (i.e., cake cutting) setting to show that a constrained serial dictatorship mechanism violates envy-freeness.}
\begin{table}[h]
	\centering
	\begin{tabular}{c|ccc}
		& $g_1$ & $g_2$ & $g_3$\\
		\hline
		$a_1$ & $1+\varepsilon$ & $1$ & $1-4\varepsilon$\\
		$a_2$ & $1-4\varepsilon$ & $1+\varepsilon$ & $1$\\
		$a_3$ & $1+\varepsilon$ & $1-4\varepsilon$ & $1$
	\end{tabular}
\end{table}

The allocations produced by various agent orderings are as follows: The order $(a_1,a_2,a_3)$ induces the allocation $A^1 = (\{g_1\},\{g_2\},\{g_3\})$, the order $(a_1,a_3,a_2)$ induces the allocation $A^2 = (\{g_1\},\{g_2\},\{g_3\})$, the order $(a_2,a_1,a_3)$ induces the allocation $A^3 = (\{g_1\},\{g_2\},\{g_3\})$, the order $(a_2,a_3,a_1)$ induces the allocation $A^4 = (\{g_3\},\{g_2\},\{g_1\})$, the order $(a_3,a_1,a_2)$ induces the allocation $A^5 = (\{g_2\},\{g_3\},\{g_1\})$, and the order $(a_3,a_2,a_1)$ induces the allocation $A^6 = (\{g_3\},\{g_2\},\{g_1\})$.

Let $A \coloneqq \frac{1}{6} (A^1 + \dots + A^6)$ denote the fractional allocation implemented by randomized round-robin. Notice that $v_1(A_1) = 3(1+\varepsilon)+1+2(1-4\varepsilon) = 6-5\varepsilon$ and $v_1(A_2) = 5+(1-4\varepsilon) = 6-4\varepsilon$, implying that $a_1$ envies $a_2$.
\label{eg:RR_violates_ex-ante-EF}
\end{example}

Let us revisit our motivating example from the introduction, where the same set of goods need to be allocated repeatedly over time. As we discussed, a seemingly promising solution is to use round-robin every time, and cycle through all $n!$ possible agent orderings to eliminate the ``advantage'' that any agent has over another in a specific ordering. In the long run, the allocation produced would resemble the fractional allocation returned by randomized round-robin. While \Cref{prop:RRR-Prop} shows that this would guarantee each agent her proportional share in the long run, \Cref{eg:RR_violates_ex-ante-EF} points out that envy may not vanish, even in the long run. 

Instead of starting from a method that guarantees ex-post EF1 and using it to achieve ex-ante EF, let us do the opposite: Start from a fractional EF allocation and implement it using integral EF1 allocations. \emph{Probabilistic serial} is a well-studied algorithm that produces a fractional envy-free allocation~\citep{BM01}. This algorithm uses a \emph{serial eating} protocol wherein all agents simultaneously start eating their respective favorite goods at the same constant speed. Once a good is completely consumed by a subset of agents, each of those agents proceeds to eating her favorite available good at the same speed. The algorithm terminates when all goods have been eaten, and the fraction of each good consumed by an agent is allocated to her. A useful property of this algorithm is that it only uses the ordinal preferences of agents over goods, and computes an allocation that is ex-ante envy-free for \emph{any} additive utilities consistent with the ordinal preferences; thus, it is \emph{SD-envy-free}.


In \Cref{subsec:RPS}, we present a variant of this procedure, which we call \emph{recursive probabilistic serial}, that provides the desired fairness guarantees, i.e., ex-ante EF and ex-post EF1. While this procedure efficiently produces a sample integral allocation, computing the full randomized allocation it returns may take exponential time. Later, in \Cref{subsec:RPS_Polytime}, we describe a modification of this procedure that runs in strongly polynomial time. 

It is worth noting that in a follow-up to our work, Aziz~\cite{Aziz20} has shown that the fractional allocation produced by probabilistic serial can also be implemented using an ex-post EF1 randomized allocation. The key idea in his proof is to consider multiple ``copies'' of each agent---one for each unit of time---such that the $i$-th copy consumes the same good(s) as consumed by the agent between the time steps $t = i-1$ and $t=i$. The desired implementation is then obtained by applying Birkhoff-von Neumann theorem to the ``expanded'' allocation matrix whose columns comprise of the goods, and rows comprise of all copies of each agent.

\subsection{Recursive Probabilistic Serial}
\label{subsec:RPS}

The eating protocol used by the probabilistic serial algorithm can be formally defined as follows: Given as input a subset of goods $M' \subseteq M$, the algorithm allows each agent to consume its favorite available good in $M'$ out of those that have not yet been fully consumed at a rate of one good per unit time. Let $X^t \coloneqq \textsc{Eating}(M',t)$ denote the (partial) allocation obtained when the algorithm is run for $t$ units of time. Note that \textsc{Eating}$(M,\lceil \frac{m}{n}\rceil)$ is equivalent to the outcome of the probabilistic serial algorithm~\citep{BM01}.

We will make use of the following classic theorem~\cite{Birk46,Neu53}, stated in a general form.

\begin{theorem}[Birkhoff-von Neumann]\label{thm:BvN}
	Let $X$ be a real-valued $n \times m$ matrix such that
	\begin{enumerate}
		\item $X_{ij} \in [0,1]$ for every $i \in [n]$ and $j \in [m]$,
		\item $\sum_{j=1}^m X_{ij} \le 1$ for every $i \in [n]$, and
		\item $\sum_{i=1}^n X_{ij} \le 1$ for every $j \in [m]$.
	\end{enumerate}
	Then, in strongly polynomial time, one can compute integral matrices $A^1, A^2, \ldots, A^q$ and weights $w^1, w^2, \ldots, w^q \in [0,1]$ such that 
	\begin{enumerate}
	\item $\sum_{k=1}^q w^k=1$ and $X = \sum_{k=1}^q w^k A^k$,
	\item for each integral matrix $A^k$, we have $A_{ij}^k \in \set{0,1}$ for every $i\in[n]$ and $j\in[m]$, and
	\item for every $k \in [q]$, we have
	\begin{enumerate}
		\item $\floor{\sum_{j=1}^m X_{ij}} \le \sum_{j=1}^m A^k_{ij} \le \ceil{\sum_{j=1}^m X_{ij}}$ for every $i \in [n]$, and
		\item $\floor{\sum_{i=1}^n X_{ij}} \le \sum_{i=1}^n A^k_{ij} \le \ceil{\sum_{i=1}^n X_{ij}}$ for every $j \in [m]$.
	\end{enumerate}
	\end{enumerate}
\end{theorem}

We will now describe our algorithm. The algorithm proceeds in stages. In each stage, the \textsc{Eating} procedure is run for \emph{one} unit of time, after which every agent has consumed a total mass of one item (although an individual item may be fractionally consumed). Then, \Cref{thm:BvN} is used to decompose this fractional (partial) allocation into integral allocations. By implication (3a) of \Cref{thm:BvN}, each agent receives exactly one item in each integral allocation in the support. We interpret the weights $w^1,\dots,w^q$ as probabilities, and sample an integral allocation from this distribution before proceeding to the next stage. The algorithm then fixes the assignments of the goods to the agents in accordance with the aforementioned integral partial allocation, and repeats the eating procedure with the reduced set of goods. In the final stage of the algorithm, fewer than $n$ goods might remain, in which case some agents do not receive any good under the integral allocation. The full procedure is described in Algorithm~\ref{alg:ef-ef1}. Note that Algorithm~\ref{alg:ef-ef1} is a randomized algorithm that returns an integral allocation in every execution. One could, however, also consider the implied fractional allocation, which is obtained by taking expectation over the randomness in Line~\ref{line:random}. Algorithm~\ref{alg:RPS} in \Cref{sec:RPS_recursive_formulation} provides the pseudocode for a recursive formulation of the algorithm. \Cref{eg:Rec_Prob_Serial} illustrates the execution of our algorithm.

\begin{algorithm}[t]
	\caption{Ex-ante EF + ex-post EF1}
	\label{alg:ef-ef1}
	\begin{algorithmic}[1]
		\State $B \leftarrow (\emptyset,\dots,\emptyset)$
		\Comment{Initialize the allocation}
		\State $M^0 \leftarrow M$
		\Comment{Initialize the set of available goods}
		\For{$t=1$ to $\lceil \frac{m}{n} \rceil$}
		\State $X^t \leftarrow \textsc{Eating}(M^t,1)$
		\Comment{Run the eating protocol for one unit of time}
		\State BvN: $X^t = \sum_{k=1}^\ell w^{t,k} A^{t,k}$
		\Comment{Decomposition from \Cref{thm:BvN}}
		\State $B^t \leftarrow \ A^{t,k}$ with probability $w^{t,k}$
		\Comment{Sampling step}\label{line:random}
		\State $M^{t+1} \leftarrow M^t \setminus \bigcup_{i \in N} B^t_i$
		\Comment{Fix the assignments according to the sampled allocation}
		 \State $B_i \leftarrow B_i \cup_{t=1}^{\lceil \frac{m}{n} \rceil} B_i^t$ for all $i \in N$
		 \Comment{Update integral allocation}
		\EndFor
		\State \Return $B$
	\end{algorithmic}
\end{algorithm}

\begin{example}
\label{eg:Rec_Prob_Serial}
Consider an instance with four goods and two agents $a_1,a_2$. The preferences of $a_1$ and $a_2$ are such that $v_1(g_1)>v_1(g_2)>v_1(g_3)>v_1(g_4)$ and $v_1(g_1)>v_1(g_3)>v_1(g_2)>v_1(g_4)$. In the first stage of the algorithm, $a_1$ consumes half of $g_1$ and half of $g_2$, while $a_2$ consumes half of $g_1$ and half of $g_3$. \Cref{fig:Rec_Prob_Serial} shows the corresponding fractional allocation and its decomposition. 

\begin{figure}[t]
	\centering
\tikzset{every picture/.style={line width=0.3pt}}
\begin{tikzpicture}
\footnotesize
\def\x{-3};
        \matrix [matrix of math nodes,left delimiter=(,right delimiter=)] (m) at (\x+3,0) 
        {
            0.5 & { } & 0.5 & { } & 0 & { } & 0 \\ 
            0.5 & { } & 0 & { } & 0.5 & { } & 0 \\ 
        };
        \node[above=0.05cm of m] {$X^1$};
%
		\node [right=0.3cm of m] (eq) {$=$};
%
		\node [right=0cm of eq] (mult1) {$0.5\ \times\ $};
%
        \matrix [matrix of math nodes,left delimiter={[},right delimiter={]},right=0.1cm of mult1] (m1)  
        {
            1 & { } & 0 & { } & 0 & { } & 0 \\  
            0 & { } & 0 & { } & 1 & { } & 0 \\ 
        };
        \node[above=0.1cm of m1] {$A^1$};
%
		\node [right=0.3cm of m1] (plus) {$+$};
%
		\node [right=0cm of plus] (mult2) {$0.5\ \times\ $};
%
        \matrix [matrix of math nodes,left delimiter={[},right delimiter={]},right=0.1cm of mult2] (m2)
        {
            0 & { } & 1 & { } & 0 & { } & 0 \\  
            1 & { } & 0 & { } & 0 & { } & 0 \\ 
        };
        \node[above=0.1cm of m2] {$A^2$};
    \end{tikzpicture}
    \caption{The fractional allocation $X^1$ and its decomposition used in \Cref{eg:Rec_Prob_Serial}.}
    \label{fig:Rec_Prob_Serial}
	\vspace{-0.1in}
\end{figure}

Suppose the integral allocation sampled by the algorithm is $A^1$ (this happens with probability 0.5). Thus, $a_1$ is allocated $g_1$ and $a_2$ is allocated $g_3$, leaving $g_2$ and $g_4$ remaining in the second stage. Since both agents prefer $g_2$ to $g_4$, each agent consumes half of each good. Therefore, conditioned on allocation $A^1$ being sampled in the first stage, each agent is allocated each of $g_2$ and $g_4$ with probability 0.5 at the end of the second stage. The corresponding distribution is given by
\begin{align*}
0.5 \times
\begin{bmatrix}
1 & 0 & 0 & 1 \\
0 & 1 & 1 & 0 \\
\end{bmatrix}
+
0.5 \times
\begin{bmatrix}
1 & 1 & 0 & 0 \\
0 & 0 & 1 & 1 \\
\end{bmatrix}.
\end{align*}
Similarly, if allocation $A^2$ is sampled at the end of first stage, then each of the two agents gets half of $g_3$ and $g_4$, resulting in the following (conditional) distribution at the end of second stage:
\begin{align*}
0.5 \times
\begin{bmatrix}
0 & 1 & 1 & 0 \\
1 & 0 & 0 & 1 \\
\end{bmatrix}
+
0.5 \times
\begin{bmatrix}
0 & 1 & 0 & 1 \\
1 & 0 & 1 & 0 \\
\end{bmatrix}.
\end{align*}

Combining the above observations, we get that the distribution over integral allocations produced by Algorithm~\ref{alg:ef-ef1} on this example is
\begin{align*}
0.25 \times
\begin{bmatrix}
1 & 0 & 0 & 1 \\
0 & 1 & 1 & 0 \\
\end{bmatrix}
+
0.25 \times 
\begin{bmatrix}
1 & 1 & 0 & 0 \\
0 & 0 & 1 & 1 \\
\end{bmatrix}
+
0.25 \times
\begin{bmatrix}
0 & 1 & 1 & 0 \\
1 & 0 & 0 & 1 \\
\end{bmatrix}
+
0.25 \times
\begin{bmatrix}
0 & 1 & 0 & 1 \\
1 & 0 & 1 & 0 \\
\end{bmatrix}.
\end{align*}

Observe that the fractional allocation returned by Algorithm~\ref{alg:ef-ef1} differs from that returned by the probabilistic serial algorithm on this example. Indeed, the latter assigns $g_2$ exclusively to $a_1$ and $g_3$ exclusively to $a_2$, whereas Algorithm~\ref{alg:ef-ef1} does not. As a consequence, Algorithm~\ref{alg:ef-ef1} fails to satisfy ordinal efficiency, which is satisfied by probabilistic serial.
\end{example}

We will first show that every integral allocation returned by Algorithm~\ref{alg:ef-ef1} satisfies EF1. The proof relies crucially on the following lemma.

\begin{lemma}\label{lem:sd-ef1-lemma}
	Let $t < \ceil{m/n}$. Let $g_{i,t}$ denote the good allocated to agent $i$ in iteration $t$ of Algorithm~\ref{alg:ef-ef1}. Then agent $i$ (weakly) prefers $g_{i,t}$ to all goods in $M^{t+1}$ (goods unallocated after $t$ iterations of the algorithm).
\end{lemma}
\begin{proof}
Suppose, for contradiction, that there exists a good $g^* \in M^{t+1}$ such that $v_i(g^*) > v_i(g_{i,t})$. Then, by \Cref{thm:BvN}, agent $i$ must have eaten a non-zero share of good $g_{i,t}$ in iteration $t$. On the other hand, $g^* \in M^{t+1}$ implies that $g^* \in M^t$, and thus $g^*$ must not be fully consumed in iteration $t$ (i.e., $\sum_{s=1}^t \sum_{i=1}^n B^s_{i,{g^*}} < 1$). However, since $v_i(g^*) > v_i(g)$, agent $i$ should have then consumed $g^*$ before consuming $g_{i,t}$ in iteration $t$, which is a contradiction. 
\end{proof}

\begin{lemma}\label{lem:ef-ef1}
	Every integral allocation returned by Algorithm~\ref{alg:ef-ef1} satisfies EF1.
\end{lemma}
\begin{proof}
	For any $i \in [n]$, let $g_{i,t}$ denote the good received by agent $i$ in iteration $t$ (if such a good exists). Fix a pair of agents $i,k \in N$. By \Cref{lem:sd-ef1-lemma}, for every $t < \ceil{m/n}$, we have $v_i(g_{i,t}) \ge v_i(g_{k,t+1})$ because $g_{k,t+1} \in M^{t+1}$. Hence, 
\begin{align*}
	v_i(A_i) \ge \sum_{t=1}^{\ceil{m/n}-1} v_i(g_{i,t}) \ge \sum_{t=2}^{\ceil{m/n}} v_i(g_{k,t}) = v_i(A_k \setminus \set{g_{k,1}}),
\end{align*}
which establishes EF1.
\end{proof}

\newcommand{\E}{\mathbb{E}}
\begin{theorem}\label{thm:eating_ef}
	The randomized allocation implemented by Algorithm~\ref{alg:ef-ef1} is ex-ante envy-free and ex-post EF1.
\end{theorem}
\begin{proof}
The ex-post EF1 guarantee follows readily from \Cref{lem:ef-ef1}. To prove ex-ante EF, notice that each agent eats her most preferred available good at each point in time during the \textsc{Eating} algorithm. Therefore, for any iteration $t$, the partial fractional allocation $X^t$ computed by the \textsc{Eating}$(\cdot)$ procedure in iteration $t$ is envy-free.

Fix a pair of agents $i,h \in N$. We want to show that $\E[v_i(B_i)] \ge \E[v_i(B_h)]$. Note that $B_i = \bigcup_{t=1}^{\ceil{m/n}} B^t_i$ and $B_h = \bigcup_{t=1}^{\ceil{m/n}} B^t_h$. Due to additive preferences and linearity of expectation, it is sufficient to show that for every $t$, $\E[v_i(B^t_i)] \ge \E[v_i(B^t_h)]$. Note that conditioned on $B^1,\ldots,B^{t-1}$, the expected value of any agent for $B^t_i$ is the value of that agent for $X^t_i$. Hence, for each possible value of $B^1,\ldots,B^{t-1}$, envy-freeness of $X^t$ implies that
$$\E[v_i(B^t_i) | B^1,\ldots,B^{t-1}] = \E[v_i(X^t_i) | B^1,\ldots,B^{t-1}] \ge
\E[v_i(X^t_h) | B^1,\ldots,B^{t-1}] =  \E[v_i(B^t_h) | B^1,\ldots,B^{t-1}],$$
as desired.
\end{proof}

Note that like probabilistic serial, Recursive Probabilistic Serial (\Cref{alg:ef-ef1}) takes only the ordinal preferences of agents over goods as input; thus, it is in fact ex-ante SD-EF and ex-post SD-EF1. 

\subsection{A Polynomial-Time Algorithm for Ex-ante EF and Ex-post EF1}
\label{subsec:RPS_Polytime}

Note that recursive probabilistic serial (Algorithm~\ref{alg:ef-ef1}) efficiently produces a sample integral EF1 allocation, but computing the full randomized allocation constructed by it would require exponential time. This is because at each iteration, the algorithm ``branches out'' into a polynomial number of subinstances (equal to the size of the support in the BvN decomposition) and this happens polynomially many times. Also, for the same reason, the randomized allocation it constructs can, in the worst case, have an exponential number of integral allocations in the support. 

A natural approach towards reducing the support size (in the hope of reducing the overall running time) is to use the well-known Carath{\'e}odory's theorem from convex analysis, which states that if a point $x \in \mathbb{R}^d$ lies in the convex hull of a set $P$, then $x$ can be expressed as a convex combination of at most $d+1$ vertices of $P$. Since any fractional allocation $X$ is a point in an $n\times m$-dimensional space (i.e., $X \in \mathbb{R}^{n \times m}$), by Carath{\'e}odory's theorem, if it is a convex combination of integral allocations in the support $\set{A^1,\ldots,A^{\ell}}$, then it can also be written as a convex combination of at most $mn+1$ allocations from this set. While this establishes the \emph{existence} of a convex combination with small support, it does not automatically provide an efficient \emph{algorithm} for converting a convex combination with large support to one with small support. Indeed, reducing the support size to at most $mn+1$ once it has already grown exponentially may not be possible efficiently. Instead, in our modification of recursive probabilistic serial, we ensure that the support size remains $mn+1$ after every step of ``branching out''. Note that every step of ``branching out'' can increase the support size from $mn+1$ to $(mn+1)\cdot p(m,n)$, where $p(m,n)$ is a polynomial that bounds the size of the support produced by the BvN decomposition. Thus, we need a way to reduce the support size from at most $(mn+1)\cdot p(m,n)$ to at most $mn+1$. This is where the following beautiful result comes in handy. 

\begin{proposition}[\citet{GLS81ellipsoid}]\label{lem:GLS}
	There is a polynomial-time algorithm that, given a polytope $\mathcal{P}$ with a strong separation oracle and a point $y \in \mathcal{P}$, returns a representation of $y$ as a convex combination of $\textup{dim}(\mathcal{P})+1$ vertices of $\mathcal{P}$.
\end{proposition}

We refer the reader to their work for formal definitions of a polytope, its dimension, and a strong separation oracle. Let us describe in a bit more detail how we use this result. It is instructive to revisit the first two iterations of \Cref{alg:ef-ef1}. Notice that at the end of the first iteration ($t=1$), the fractional partial allocation $X^1$ consists of a polynomial number of integral allocations in its support, i.e., $X^1 = \sum_{k=1}^\ell w^{1,k} A^{1,k}$, where $\ell = p(m,n)$ for some polynomial $p$. Consider the polytope $\mathcal{P}^1$ formed by the convex hull of the integral allocations $A^{1,1},\dots,A^{1,\ell}$ as follows:
$$\mathcal{P}^1 \coloneqq \{z : z = \sum_k \alpha_{k}A^{1,k} \text{ where } \alpha_k \geq 0 \text{ for all } k \in [\ell] \text{ and } \sum_k \alpha_k=1\}.$$
Observe that $\mathcal{P}^1$ admits a polynomial-time strong separation oracle. That is, one can decide in polynomial time whether a given point $y \in \mathbb{R}^{n \times m}$ belongs to $\mathcal{P}^1$, or return a separating hyperplane if it doesn't (this can be done via linear programming). This implies that an $(mn+1)$-sized decomposition of $X^1$, as guaranteed by Carath{\'e}odory's theorem, can be computed in polynomial time.

Having computed an $(mn+1)$-sized decomposition of $X^1$, we now obtain at most $mn+1$ subinstances in each of which we re-run our algorithm. Notice that doing so results in a probability distribution over $(mn+1) \cdot p(m,n)$ deterministic allocations at the end of the second iteration. We can now re-invoke the aforementioned algorithmic version of Carath{\'e}odory's theorem to once again obtain a decomposition over at most $mn+1$ of these integral allocations in polynomial time. It is easy to see that this argument in fact holds for all time steps via induction. 

Notice that the randomized allocation thus produced has support that is a subset of the support of the randomized allocation produced by recursive probabilistic serial. Hence, ex-post EF1 of this modified procedure follows immediately from \Cref{lem:sd-ef1-lemma}. Finally, it is also easy to verify that the randomized allocation is also ex-ante EF by essentially the same argument as in the proof of \Cref{thm:eating_ef}: each individual step of the eating procedure produces an envy-free allocation, so even if we change the convex combination over such allocations, it remains envy-free.

Again, like Recursive Probabilistic Serial (\Cref{alg:ef-ef1}), this algorithm takes only the ordinal preferences of agents over goods as input; thus, it is in fact ex-ante SD-EF and ex-post SD-EF1. 

\section{Impossibility: Ex-ante Prop + ex-post EF1 + ex-post fPO}
\label{sec:Impossibility}

In the previous section, we showed that ex-ante EF and ex-post EF1 can be achieved simultaneously in polynomial time. The obvious next question, then, is whether we can achieve stronger guarantees. One property that our algorithms from Section~\ref{sec:EF_EF1} lack is efficiency. We defined three efficiency notions in Section~\ref{sec:Preliminaries} that are related through the following logical implications: ex-ante PO $\Rightarrow$ ex-post fPO $\Rightarrow$ ex-post PO. 

Let us consider adding the weakest of them: ex-post PO. Unfortunately, we were not able to settle whether ex-ante EF (or even ex-ante Prop) is compatible with ex-post EF1 and ex-post PO. The difficulty is that we do not understand the structure of integral EF1+PO allocations; we elaborate on this issue in Section~\ref{sec:Discussion} and lay out several more open questions that stem from this. 

Let us move on and consider imposing a slightly stronger efficiency notion: ex-post fPO. We note that integral EF1+fPO allocations are known to always exist~\cite{BKV18}. Hence, the question of whether we can randomize over such allocations to achieve a desirable ex-ante fairness guarantee is meaningful. However, in this case, we show that even achieving ex-ante proportionality is impossible along with ex-post EF1 and ex-post fPO. 

\begin{theorem}\label{thm:impossible-prop-ef1-fpo}
	There exists an instance with additive valuations in which no randomized allocation is simultaneously ex-ante proportional, ex-post envy-free up to one good, and ex-post fractionally Pareto optimal.
\end{theorem}
\begin{proof}
We present an instance in which the unique integral allocation satisfying EF1+fPO violates proportionality. Specifically, consider an instance with two goods ($g_1,g_2$) and two agents ($a_1,a_2$) whose additive valuations are as follows: $v_{1,1}=1,v_{1,2}=2$ and $v_{2,1}=1,v_{2,2}=3$. 
This instance has exactly two integral EF1 allocations: $A \coloneqq (\{g_1\},\{g_2\})$ and $B \coloneqq (\{g_2\},\{g_1\})$. 
It is easy to check that $B$ is not fPO, since it is Pareto dominated by a fractional allocation $X$ that assigns $g_1$ completely to $a_1$, and splits $g_2$ equally between the two agents. Indeed, $v_1(X_1) = v_1(g_1) + 0.5 \cdot v_1(g_2) = 2 \geq v_1(B_1)$ and $v_2(X_2) = 0.5 \cdot v_2(g_2) = 1.5 > v_2(B_2)$. To see why $A$ is fPO, notice that it assigns each good to an agent that has the highest valuation for it. Therefore, $A$ maximizes the utilitarian social welfare (i.e., sum of agents' utilities) and is therefore fPO. 

Finally, observe that $A$ violates proportionality since $v_1(A_1) = v_1(g_1) < \frac{1}{2} v_1(g_1 \cup g_2)$. Therefore, any randomized allocation that is ex-post EF1+fPO must be supported entirely on the integral allocation $A$, in which case it violates ex-ante Prop, as desired.
\end{proof}


For example, since ex-ante envy-freeness and ex-ante group fairness imply ex-ante proportionality, they are also incompatible with ex-post EF1+fPO. Similarly, since ex-ante PO implies ex-post fPO (\Cref{prop:po-fpo}), we also get that ex-ante Prop+PO and ex-ante EF+PO are incompatible with ex-post EF1. 


We mentioned earlier that achieving ex-ante Prop or ex-ante EF along with ex-post EF1+PO is an open question. Two prominent methods for finding an integral EF1+PO allocation are the  integral MNW rule~\cite{CKMP+19} and the market-based rule of \citet{BKV18}. An interesting implication of \Cref{thm:impossible-prop-ef1-fpo} is that we cannot hope to achieve even ex-ante proportionality by randomizing over allocations returned by either method. The latter method is guaranteed to return an integral EF1+fPO allocation, so \Cref{thm:impossible-prop-ef1-fpo} directly applies. And the MNW rule, while only guaranteed to return an integral EF1+PO allocation, uniquely returns allocation $A$ in the example presented in the proof of \Cref{thm:impossible-prop-ef1-fpo}, which violates proportionality.

\section{Possibility: Ex-ante GF + Ex-post Prop1 + Ex-Post EF$_1^1$}
\label{sec:Prop1_EF11}

Given the impossibility of ex-ante Prop and ex-post EF1+fPO derived in the previous section, it is evident that retaining efficiency would require relaxing at least one of the fairness guarantees. Ex-ante Prop is already a very weak fairness guarantee. So, in this section, we focus on relaxing ex-post EF1. There are two relaxations that have been proposed in the literature --- namely, Prop1 and EF$_1^1$. We show that both of these can be achieved simultaneously, and in fact, this can be done while strengthening ex-ante Prop and ex-post fPO to ex-ante GF. Recall that ex-ante GF implies not only ex-ante Prop, but also ex-ante PO, which, by Proposition~\ref{prop:po-fpo}, implies ex-post fPO for any implementation of it. In other words, our goal is to implement a fractional GF allocation using integral Prop1+EF$_1^1$ allocations. 

Luckily, we know that a fractional GF allocation always exists. \citet{CFSV19} argued that a fractional MNW allocation satisfies group fairness. They did not provide a formal proof; for completeness, we provide a proof of this result for the more general cake-cutting setting in the appendix. Further, we know that a fractional MNW allocation can be computed in strongly polynomial time~\cite{Orlin10,Vegh16}.

Hence, we ask whether a fractional MNW allocation can be implemented using integral Prop1+ EF$_1^1$ allocations. Our starting point is a result by \citet{BCKM13} that allows implementing any fractional allocation using integral allocations that are very ``close'' to it in agent utilities. Specifically, they prove the next result deriving and using an extension of the classic Birkhoff-von Neumann theorem~\cite{Birk46,Neu53}. 
 
\begin{proposition}[Utility Guarantee; Theorem 9 of \citet{BCKM13}]\label{prop:BCKMutility}
Given any fractional allocation $X$, one can compute, in strongly polynomial time, a randomized allocation implementing $X$ whose support consists of integral allocations $A^1,\dots,A^\ell$ such that for every $k \in [\ell]$ and every agent $i \in N$,
\begin{align*}
|v_i(X_i) - v_i(A^k_i)| \leq \max \{v_i(g) - v_i(g') : 0 < X_{i,g},X_{i,g'} < 1 \text{ and } g,g' \in M \}.
\end{align*}
\end{proposition}

Notice that the upper bound established in \Cref{prop:BCKMutility} on how much agent $i$'s utility under an integral allocation $A^k$ in the support can differ from her utility under the fractional allocation $X$ depends only on the fractional allocation $X$. In contrast, the fairness guarantees we want to establish for the integral allocations in the support --- Prop1 and EF$_1^1$ --- consider what happens when we add a good to the bundle of agent $i$ that agent $i$ \emph{is not already allocated}; in other words, we need a stronger guarantee for integral allocations in the support which depends on which goods the agent is (or is not) allocated ex-post. 

It turns out that the method proposed by \citet{BCKM13} already provides such a guarantee, and their proof can be adapted to establish a more nuanced bound. Specifically, we show that if the agent's ex-ante utility $v_i(X_i)$ exceeds her ex-post utility $v_i(A^k_i)$, then the gap is at most the maximum value the agent has for any good that she \emph{lost} in the integral allocation (i.e., any good $g$ such that $0 < X_{i,g} < 1$ and $A^k_{i,g} = 0$). Similarly, if the ex-post utility exceeds the ex-ante utility, then the gap is at most the maximum value the agent has for any good that she \emph{gained} in the integral allocation (i.e., any good $g$ such that $0 < X_{i,g} < 1$ and $A^k_{i,g} = 1$). We later show that this subtle improvement helps us establish the desired ex-post fairness guarantees.

\begin{lemma}[Utility Guarantee++]\label{lem:utility}
	Given a fractional allocation $X$, one can compute, in strongly polynomial time, a randomized allocation implementing $X$ whose support consists of integral allocations $A^1,\ldots,A^{\ell}$ such that for every $k \in [\ell]$ and every agent $i \in N$, the following hold:
	\begin{enumerate}
		\item If $v_i(A^k_i) < v_i(X_i)$, then $\exists \, g_i^{-} \notin A^k_i$ with $X_{i,g_i^{-}} > 0$ such that $v_i(A^k_i) + v_i(g_i^{-}) > v_i(X_i)$.
		\item If $v_i(A^k_i) > v_i(X_i)$, then $\exists \, g_i^{+} \in A^k_i$ with $X_{i,g_i^{+}} < 1$ such that $v_i(A^k_i)- v_i(g_i^{+}) < v_i(X_i)$. 
	\end{enumerate}
\end{lemma}
\begin{proof}
In their proof of \Cref{prop:BCKMutility}, \citet{BCKM13} propose the following method for computing an implementation of a given fractional allocation $X$. Consider a fixed agent $i \in N$. Suppose the goods in $M$ are indexed as $g_{i,1},\dots,g_{i,m}$ so that $v_i(g_{i,k}) \ge v_i(g_{i,k+1})$ for each $k \in [m-1]$. For simplicity, we will write $v_{i,k} \coloneq v_i(g_{i,k})$ for all $k \in [m]$ and $v_{i,m+1} \coloneq 0$.

For any $i \in N$ and any $k \in [m]$, define $Q_{i,k} \coloneqq \sum_{t=1}^k X_{i,g_{i,t}}$ as the total fractional amount of the $k$ most preferred goods assigned to agent $i$ under $X$. Consider the following set of constraints on a generic fractional allocation $Y$: 
\begin{equation}\label{eqn:constraints}
\begin{aligned}
&\mathcal{H}_1 : \floor{Q_{i,k}} \le \sum_{t=1}^k Y_{i,g_{i,t}} \le \ceil{Q_{i,k}},\ \forall i \in N \text{ and } \forall k \in [m],\\
&\mathcal{H}_2 : \sum_{i \in N} Y_{i,g} = 1,\ \forall g \in M.
\end{aligned}
\end{equation}
Observe that $X$ trivially satisfies these constraints. \citeauthor{BCKM13} show that these constraints have the so-called ``bihierarchy'' structure (refer to \Cref{sec:BCKMRounding} for a formal definition), which allows computing, in strongly polynomial time, an implementation of $X$ whose support consists of integral allocations $A^1,\ldots,A^{\ell}$ that also satisfy these constraints. 

Lastly, \citeauthor{BCKM13} show that any integral allocation satisfying the constraints in \Cref{eqn:constraints} must satisfy the guarantee in \Cref{prop:BCKMutility}. We show that it in fact satisfies the slightly stronger guarantee that we seek. For simplicity, let us write $\hat{A}$ to denote a generic integral allocation satisfying the constraints in \Cref{eqn:constraints}, and $\hat{Q}_{i,k} \coloneqq \sum_{t=1}^k \hat{A}_{i,g_{i,t}}$ for all $i \in N$ and $k \in [m]$.

Let us first analyze the case where $v_i(\hat{A}_i) < v_i(X_i)$. Then, there must exist some good $g \in M$ such that $\hat{A}_{i,g} < X_{i,g}$. Since $\hat{A}_{i,g} \in \set{0,1}$ and $X_{i,g} \in [0,1]$, this is equivalent to $\hat{A}_{i,g} = 0 < X_{i,g}$. Let $k^{-}$ be the smallest index such that $\hat{A}_{i,g_{i,k^{-}}} < X_{i,g_{i,k^{-}}}$, i.e., $g_{i,k^{-}}$ is agent $i$'s most preferred good satisfying this condition. Hence, $g_{i,k^{-}} \notin \hat{A}_i$ and $X_{i,g_{i,k^{-}}} > 0$. Further, for all $k < k^{-}$, we have $\hat{A}_{i,g_{i,k}} \ge X_{i,g_{i,k}}$, and, as a result, $\hat{Q}_{i,k} \ge Q_{i,k}$. Thus, 
	\begin{align*}
	v_i(X_i) - v_i(\hat{A}_i) &= \sum_{k=1}^m v_{i,k} \cdot (X_{i,g_{i,k}} - \hat{A}_{i,g_{i,k}}) = \sum_{k=1}^m (v_{i,k}-v_{i,k+1}) \cdot (Q_{i,k} - \hat{Q}_{i,k})\\
	&\le \sum_{k=k^{-}}^m (v_{i,k}-v_{i,k+1}) \cdot (Q_{i,k} - \hat{Q}_{i,k}) < \sum_{k=k^{-}}^m (v_{i,k}-v_{i,k+1}) \cdot 1 = v_{i,k^{-}},
	\end{align*}
	where the second transition is a simple algebraic exercise, the third transition holds because we noted that $\hat{Q}_{i,k} \ge Q_{i,k}$ for all $k < k^{-}$, and the fourth transition holds because $\hat{A}$ satisfies $\H_1$ in \Cref{eqn:constraints}, and therefore, we have that $Q_{i,k} - \hat{Q}_{i,k} \le Q_{i,k} - \floor{Q_{i,k}} < 1$. Taking $g_i^{-} \coloneqq g_{i,k^{-}}$, we notice that this is the guarantee we desire when $v_i(\hat{A}_i) < v_i(X_i)$. 
	
	Next, consider the other case where $v_i(\hat{A}_i) > v_i(X_i)$. Then, there must exist some good $g \in M$ such that $\hat{A}_{i,g} > X_{i,g}$. Note that this is equivalent to $\hat{A}_{i,g} = 1 > X_{i,g}$. Let $k^{+}$ be the smallest index such that $g_{i,k^{+}}$ satisfies this condition. Then, we have that $g_{i,k^{+}} \in \hat{A}_i$, $X_{i,g_{i,k^{+}}} < 1$, and by an argument similar to the one above, $v_i(\hat{A}_i) - v_i(X_i) < v_{i,k^{+}}$. Hence, in this case, we can take $g_i^{+} \coloneqq g_{i,k^{+}}$, as desired.	
\end{proof}

We now show how \Cref{lem:utility} can be used to achieve our desired ex-post fairness guarantees of Prop1 and EF$_1^1$. 

\begin{theorem}\label{thm:prop1-ef11}
	There is a strongly polynomial-time algorithm that, given any fractional proportional (Prop) allocation as input, computes an implementation of it using integral allocations that are proportional up to one good (Prop1). If, in addition, the input is a fractional MNW allocation, then the integral allocations in the support also satisfy envy-freeness up to one good more-and-less (EF$_1^1$). 
\end{theorem}
\begin{proof}
	Let $X$ be a fractional allocation, and let $A^1,\ldots,A^{\ell}$ be integral allocations in the support of an implementation of $X$ produced by \Cref{lem:utility}. 
	
	Suppose $X$ satisfies proportionality. We want to show that for each $k \in [\ell]$, $A^k$ is Prop1. Since $X$ is proportional, for every $i \in N$, $v_i(X_i) \ge v_i(\vec{1}^m)/n$, where $v_i(\vec{1}^m)$ is agent $i$'s utility for receiving all goods fully. Fix $k \in [\ell]$. By \Cref{lem:utility}, we have that for every agent $i \in N$, either $v_i(A^k_i) \ge v_i(X_i) \ge v_i(\vec{1}^m)/n$, or there exists a good $g \notin A^k_i$ such that $v_i(A^k_i) + v_i(g) > v_i(X_i) \ge v_i(\vec{1}^m)/n$. Therefore, $A^k$ is Prop1. 
	
	Next, suppose that $X$ maximizes the Nash social welfare among all fractional allocations. Since a fractional MNW allocation is certainly proportional~\cite{Var74}, the aforementioned argument still applies for ex-post Prop1. We show that in this case, $A^k$ is also EF$_1^1$ for each $k \in [\ell]$. Note that since $X$ is a fractional MNW allocation, the following condition is satisfied for any pair of agents $i,h \in N$ and any good $g \in M$:\footnote{The condition that transferring an arbitrarily small fraction of good $g$ from agent $i$ to agent $h$ does not increase Nash welfare reduces to this condition.} 
	\begin{equation}\label{eqn:CEEI}
	X_{i,g} > 0 \implies \frac{v_i(g)}{v_i(X_i)} \ge \frac{v_h(g)}{v_h(X_h)}.
	\end{equation}

	Fix a pair of distinct agents $i,h \in N$. By \Cref{lem:utility}, either $v_i(A^k_i) \le v_i(X_i)$, or there exists $g_i^+ \in A^k_i$ with $X_{i,g_i^{+}} < 1$ such that $v_i(A^k_i \setminus \{g_i^+\}) < v_i(X_i)$. Similarly, either $v_h(A^k_h) \ge v_h(X_h)$, or there exists $g_h^- \notin A^k_h$ with $X_{h,g_h^{-}} > 0$ such that $v_h(A^k_h \cup \{g_h^-\}) > v_h(X_h)$. To simplify the analysis, let us assume that the second condition holds in both cases.\footnote{If $v_i(A^k_i) \le v_i(X_i)$ (or $v_h(A^k_h) \ge v_h(X_h)$), we can treat $g_i^+$ (or $g_h^-$) as a dummy good with $v_i(g_i^+) = 0$ (or $v_h(g_h^-) = 0$).}
	
	By summing the right-hand side inequality in \Cref{eqn:CEEI} over all $g \in A^k_i \setminus \set{g_i^+}$, we get 
	$$
	\frac{v_h(A^k_i \setminus \set{g_i^+})}{v_h(X_h)} \le \frac{v_i(A^k_i \setminus \set{g_i^+})}{v_i(X_i)} < 1.
	$$
	Thus, $v_h(A^k_i \setminus \{g_i^+\}) < v_h(X_h) < v_h(A^k_h \cup \set{g_h^-})$, implying that $A^k$ satisfies EF$_1^1$, as desired.
\end{proof}

Notice that the proof of Theorem~\ref{thm:prop1-ef11} establishes a stronger version of Prop1 wherein an agent not receiving her proportional share gets \emph{strictly} more than her proportional share by receiving one additional good. Similarly, it also establishes a stronger version of EF$_1^1$ wherein an agent envying another agent would strictly prefer her own allocation over the other agent's allocation after adding one missing good to her bundle and removing one good from the other agent's bundle. 

\citet{BK19} recently established that integral Prop1+EF$_1^1$+fPO allocations exist and can be computed in strongly polynomial time. They rely on special-purpose techniques for rounding a fractional MNW allocation. By constrast, \Cref{thm:prop1-ef11} uses a standard technique to round a fractional MNW allocation, computes not just \emph{one} integral Prop1+EF$_1^1$+fPO allocation but rather an implementation of the fractional MNW allocation over such integral allocations, and can be applied to \emph{any} fractional Prop+PO allocation to implement it using integral Prop1+fPO allocations.\footnote{Recall that since a fractional MNW allocation is Pareto optimal, any allocation in the support of an implementation of it must be fPO by \Cref{prop:po-fpo}.}

The reason we use a fractional MNW allocation in \Cref{thm:prop1-ef11} is because it is known to be group fair (GF). This observation, the fact that a fractional MNW allocation can be computed in strongly polynomial time~\cite{Orlin10,Vegh16}, and \Cref{thm:prop1-ef11} immediately yield the main result of this section. 

\begin{corollary}\label{fgf-prop1}
	There exists a randomized allocation that is ex-ante group fair (GF), ex-post proportional up to one good (Prop1), and ex-post envy-free up to one good more-and-less (EF$_1^1$). Further, it can be computed in strongly polynomial time.  
\end{corollary}

\section{Characterizing MNW}
\label{sec:characterization}

To prove Corollary~\ref{fgf-prop1}, we utilized a rounding of a fractional MNW allocation because it is known to be group fair. One may wonder whether it is possible to use any other fractional allocation as the basis for this result. In this section, we answer this negatively by showing a fundamental connection between group fairness and the MNW allocation rule. In particular, we show that MNW, or refinements thereof, is the only allocation rule that satisfies group fairness together with \emph{replication invariance}, which states that the set of allocations returned by an allocation rule should not substantively change if the input is replicated.

\subsection{Replication Invariance}

Replication invariance has also been well-explored in economic literature on resource allocation. Intuitively, it says that if we take an instance $I$ and duplicate it to form $k$ copies, the output on the duplicated instance contains the duplicate of each possible allocation on the original instance. Recall that an allocation rule outputs a set of (tied) allocations given an instance of the fair division problem. Formally, given an instance $I = (N,M,(v_i)_{i \in N})$ and a positive integer $k \in \mathbb{N}$, let $k I$ denote an instance in which each agent and good is replicated $k$ times. For an allocation $A$ on $I$, let $k A$ denote the allocation on $k I$ in which every copy of agent $i \in N$ receives allocation $A_i$. 

\begin{definition}[Replication Invariance]
	An allocation rule $f$ satisfies \emph{replication invariance} if, for every instance $I \in \mathcal{I}$, and every allocation $A \in f(I)$,  $k A \in f(k I)$.
\end{definition}

We remarked in Section~\ref{sec:related} that replication invariance works quite well with strictly convex preferences, which is the domain in which many of the early characterization results of the MNW rule hold. As just one example, note that under strictly convex utilities, any fractional envy-free and Pareto optimal allocation must produce identical allocations to two agents with identical valuations, as observed by \citet{Var74}. To see this, suppose for contradiction that agents $i$ and $j$ have identical valuation $v$ but receive different allocations $A_i$ and $A_j$. Then, envy-freeness implies that $v(A_i) = v(A_j)$. Now, by strong convexity, $v((A_i+A_j)/2) > v(A_i) = v(A_j)$. Hence, allocation $A$ violates fractional Pareto optimality because giving each of agents $i$ and $j$ the allocation $(A_i+A_j)/2$ would make them both strictly happier. This property comes in handy when dealing with replication because we know that all replicas of an agent must receive identical allocations. In contrast, this property does not hold for additive utilities. For example, if two agents equally like two goods, giving one good fully to one agent and the other good fully to the other agent violates this property and is yet envy-free and Pareto optimal.

\subsection{Group Fairness and Replication Invariance Characterize MNW}

We have remarked that MNW satisfies group fairness. We begin this section by showing that the fractional MNW rule (which returns the set of all fractional MNW allocations) also satisfies replication invariance.

\newcommand{\mnw}{\operatorname{MNW}}
\begin{proposition}
	\label{prop:mnw-replication-invariance}
	The fractional MNW rule satisfies replication invariance.
\end{proposition}

\begin{proof}
	Suppose that it violates replication invariance. That is, for some instance $I = (N,M,(v_i)_{i \in N})$ and $k \in \mathbb{N}$, there exists a fractional allocation $Y \in \mnw(I) = \argmax_{X \in \mathcal{X}} \prod_{i \in N} v_i(X_i)$ such that $k Y \not \in \mnw(k I)$. For an agent $i \in N$ and $\ell \in [k]$, let $i_{\ell}$ denote the $\ell^\text{th}$ copy of agent $i$ in the replicated instance $k I$. Then, there exists an allocation $Z$ on $k I$ such that 
	\[ \Pi_{i \in N} \Pi_{\ell=1}^k v_i(Z_{i_\ell}) > \Pi_{i \in N} \Pi_{\ell=1}^k v_i(Y_i).\]

	For every $i \in N$ and $\ell \in [k]$, let $Z'_{i_\ell} \coloneqq Z'_i \coloneqq \frac{\sum_{\ell '=1}^k Z_{i_{\ell '}}}{k}$. That is, $Z'$ takes the total allocation to all copies of $i$ and distributes it evenly across the copies so that each copy of agent $i$ receives $Z'_i$. Because valuations are additive and each copy of $i$ has the same valuation function, it can be seen that the Nash Welfare of $Z'$ is at least as large as that of $Z$. That is,
	\[ \Pi_{i \in N} \Pi_{\ell=1}^k v_i(Z'_i) \ge \Pi_{i \in N} \Pi_{\ell=1}^k v_i(Z_{i_\ell}) > \Pi_{i \in N} \Pi_{\ell=1}^k v_i(Y_i) \Rightarrow  \Pi_{i \in N} v_i(Z'_i) > \Pi_{i \in N} v_i(Y_i),\]
	which contradicts the fact that $Y \in MNW(I)$.
\end{proof}

We next show that MNW is in fact characterized by group fairness and replication invariance, in the sense that any allocation rule that sometimes produces an allocation that does not maximize the product of valuations necessarily violates either group fairness or replication invariance.

\begin{theorem}\label{thm:mnw-characterization}
	Any allocation rule that satisfies group fairness and replication invariance is a refinement of the fractional MNW rule.
\end{theorem}

Before proving Theorem~\ref{thm:mnw-characterization}, we show a technical lemma that will be helpful.

\begin{lemma}\label{lem:local-implies-global-mnw}
	For every instance $I$, and every Pareto optimal allocation $A \not \in \mnw(I)$, there exists a pair of agents $i,j$ and $X \subseteq A_i$ with $v_i(X) \in (0,v_i(A_i))$ and $v_j(X)>0$, such that
	\begin{equation}
		\label{eq:violation}
		  v_j(X) > v_j(A_j) \cdot \frac{v_i(X)}{v_i(A_i \setminus X)}.
	\end{equation}
\end{lemma}

\begin{proof}
Let $Z \in \mnw(I)$ and let $U = \{ i \in N: v_i(A_i) > v_i(Z_i) \}$ be the set of agents with more utility in allocation $A$ than in $Z$. First, we claim that there exist agents $i \in U$ and $j \not \in U$ with $A_i \cap Z_j \neq \emptyset$. If this were not the case, then $\cup_{i \in U}A_i = \cup_{i \in U} Z_i$. Further, since $A \notin \mnw(I)$, there exists an agent $h$ with $v_h(A_h) \neq v_h(Z_h)$. And since $A$ is Pareto optimal, we must have that $U \neq \emptyset$. Hence, the allocation $B$ with $B_i=A_i$ for every $i \in U$ and $B_i=Z_i$ for every $i \not \in U$ has strictly higher Nash welfare than $Z$, a contradiction.

Fix agents $i \in U$ and $j \notin U$ such that $A_i \cap Z_j \neq \emptyset$. Let $X = \epsilon \cdot (A_i \cap Z_j)$, where $\epsilon>0$ is small enough that
\begin{equation}\label{eq:selection-of-x}
\frac{v_i(A_i \setminus X)}{v_i(Z_i)} > \frac{v_j(A_j)}{v_j(Z_j \setminus X)}, v_i(X) \in (0,v_i(A_i)), \text{ and } v_j(X) \in (0,v_j(Z_j)).
\end{equation}
It is possible to choose such an $\epsilon$ because the first inequality in \Cref{eq:selection-of-x} holds for $\epsilon=0$ and both sides are continuous in $\epsilon$. Further, we know that $v_i(X)>0$ by Pareto optimality of $A$, and that $v_j(X)>0$ by Pareto optimality of $Z$. 

We also know, by the definition of a fractional MNW allocation, that
\begin{equation}\label{eq:mnw}
\frac{v_j(X)}{v_j(Z_j \setminus X)} \ge \frac{v_i(X)}{v_i(Z_i)}.
\end{equation}
Multiplying the left and right hand sides of the inequalities in \Cref{eq:selection-of-x,eq:mnw} gives the desired result.
\end{proof}

\begin{proof}[Proof of Theorem~\ref{thm:mnw-characterization}]
	Let $f$ be an allocation rule, and suppose that there exists an instance $I$ for which there exists an allocation $A \in f(I)$ but $A \not \in \mnw(I)$. If $A$ is not Pareto optimal then it violates group fairness. So suppose $A$ is Pareto optimal. We will show that if it satisfies replication invariance, then it must violate group fairness.
	
By Lemma~\ref{lem:local-implies-global-mnw}, there exist agents $i,j$ and $X \subseteq A_i$ with $v_i(X) \in (0,v_i(A_i))$ and $v_j(X)>0$ such that
\begin{equation*}
 v_j(X) > v_j(A_j) \cdot \frac{v_i(X)}{v_i(A_i \setminus X)}.
\end{equation*}
Further, suppose that $\frac{v_i(A_i \setminus X)}{v_i(X)}$ is an integer. This is WLOG because we can just take the set of goods $X$ and scale them down until we get an integer. \Cref{eq:violation} is preserved under this scaling due to additive valuations.

Let $k=\frac{v_i(A_i \setminus X)}{v_i(X)} \in \mathbb{N}$ and consider the replicated instance $k I$. Denote the copies of agent $i$ in instance $k I$ by $i_1, \ldots, i_k$. By replication invariance, $k A \in f(k I)$. We show that $k A$ does not satisfy group fairness. Let $S$ consist of all $k$ copies of $i$ and a single copy of $j$, and $T$ consist of all $k$ copies of $i$. For every $\ell \in [k]$, let $B_{i_\ell} = A_{i_\ell} \setminus X_\ell$, where $X_\ell$ is a copy of $X$, and let $B_j=X_1 \cup \ldots \cup X_k$. For every copy $i_{\ell}$ of agent $i$, we have
\begin{align*}
\frac{|S|}{|T|} \cdot v_{i_{\ell}}(B_{i_{\ell}}) &= \frac{|S|}{|T|} \cdot v_{i_{\ell}}(A_{i_{\ell}} \setminus X_k) = \frac{\frac{v_i(A_i \setminus X)}{v_i(X)}+1}{\frac{v_i(A_i \setminus X)}{v_i(X)}} v_i(A_i \setminus X) = v_i(A_i \setminus X)+v_i(X)=v_i(A_i)
\end{align*}
and for the single copy of $j$, we have
\begin{align*}
\frac{|S|}{|T|} \cdot v_j(B_j) &= \frac{|S|}{|T|} \cdot k \cdot v_j(X) = \frac{\frac{v_i(A_i \setminus X)}{v_i(X)}+1}{\frac{v_i(A_i \setminus X)}{v_i(X)}} \cdot \frac{v_i(A_i \setminus X)}{v_i(X)}\cdot v_j(X)\\
&= \frac{v_i(A_i)}{v_i(X)} \cdot v_j(X) > \frac{v_i(A_i \setminus X)}{v_i(X)} v_j(X)>v_j(A_j),
\end{align*}
where the final inequality follows from \Cref{eq:violation}.
Group fairness is violated.
\end{proof} 

We explore this characterization further in the appendix. In particular, we investigate the implications of relaxing or strengthening group fairness. We show that Theorem~\ref{thm:mnw-characterization} ceases to hold even under a relatively mild relaxation of group fairness. We also show that if group fairness is strengthened to \emph{fractional group fairness}, which extends group fairness to allow groups that consist of fractions of agents, then Theorem~\ref{thm:mnw-characterization} holds even without requiring replication invariance.

\section{The Case of Bads}
\label{sec:bads}

In this section, we consider the case where we are dividing \emph{bads} instead of goods. We show that our positive results from Sections~\ref{sec:EF_EF1} and~\ref{sec:Prop1_EF11} still apply in this setting.

Formally, let $M$ be a set of $m$ bads. Fractional, integral, and randomized allocations are defined as in the case of goods. However, the agents now have non-positive valuation functions, i.e., $v_{i,b} \le 0$ for all agents $i \in N$ and bads $b \in M$. The definitions of proportionality, envy-freeness, (fractional) Pareto optimality, and group fairness are unchanged from the goods case. However, the definitions of approximate fairness need to be modified. Crucially, when assessing whether an allocation is approximately fair from the perspective of agent $i$, we often \emph{remove a bad} from (rather than add a good to) $i$'s bundle and/or \emph{add a bad} to (rather than remove a good from) some other agent $j$'s bundle. This causes the relaxations to differ from their goods counterparts significantly. 

\begin{definition}[Proportionality Up To One Bad (Prop1)]
	An integral allocation $A$ is \emph{proportional up to one bad} if for every agent $i \in N$, either $v_i(A_i) \geq v_i(\vec{1}^m)/n$ or there exists a bad $j \in A_i$ such that $v_i(A_i \setminus \set{j}) \geq v_i(\vec{1}^m)/n$, where $v_i(\vec{1}^m)$ is agent $i$'s valuation for receiving all bads fully.
\end{definition}

\begin{definition}[Envy-Freeness Up To $k$ Bads (EF$k$)]
	An integral allocation $A$ is \emph{envy-free up to $k$ bads} if for every pair of agents $i,h \in N$, there exists $S_i \subseteq A_i$ with $|S_i| \le k$ such that $v_i(A_i \setminus S_i) \geq v_i(A_h)$.
\end{definition}

\begin{definition}[Envy-Freeness Up To One Bad More-and-Less (EF$_1^1$)~\cite{BK19}]
	An integral allocation $A$ is \emph{envy-free up to one bad more-and-less} if for every pair of agents $i,h \in N$ such that $A_i \neq \emptyset$, we have $v_i(A_i \setminus \set{j_i}) \geq v_i(A_h \cup \set{j_h})$ for some bads $j_i \in A_i$ and $j_h \notin A_h$.
\end{definition}

Finally, for the case of bads, there is no known equivalent of the MNW allocation rule. For example, it is known that maximizing or minimizing the product of valuations (or the product of absolute valuations) leads to dramatically unfair outcomes. However, fractional CEEI allocations are still known to exist; we refer an interested reader to the work of \citet{BMSY17}. It is easy to show, as in the case of goods, that these allocations must be group fair (GF). We do not need to introduce the definition of CEEI for bads; we only need the following property, observed by \citet[Lemma 6]{BMSY17}. If $X$ is a fractional CEEI allocation of bads, then $v_i(X_i) < 0$ for all agents $i \in N$,\footnote{This assumes that each agent has a strictly negative value for at least one bad. If an agent has zero value for all bads, we can effectively remove her from the instance.} and for all agents $i,h \in N$ and bads $j \in M$,
\begin{equation}\label{eqn:ceei-bads-condition}
X_{i,j} > 0 \Rightarrow \frac{v_{i,j}}{v_i(X_i)} \le \frac{v_{h,j}}{v_h(X_h)}.
\end{equation}

Note that the quantities on both sides of the inequality are non-negative. This property is similar to that of CEEI allocations for goods that we used in the proof of Theorem~\ref{thm:bads-prop1-ef11}.

\subsection{Ex-ante EF + Ex-post EF1 for Bads}\label{sec:bads-ef-ef2}
Note that our algorithms from Section~\ref{sec:EF_EF1} --- recursive probabilistic serial (RPS) and its polynomial-time variant --- only use ordinal preferences of agents over goods. 
Hence, these algorithms naturally extend to the case of bads: Given an agent's valuation function over bads, we can construct an ordinal preference over bads in which the bads are sorted in a non-increasing order of the agent's valuation for them. 

One might wonder how the guarantees that our algorithms provide for goods translate to the case of bads. While the ex-ante EF guarantee carries over, unfortunately the ex-post EF1 guarantee enjoyed by our algorithms in the case of goods does not directly translate to the case of bads. To see this, consider an instance with two agents $1,2$ and three bads $b_1,b_2,b_3$ such that $v_i(b_j) = -j$ for all $i \in [2]$ and $j \in [3]$. That is, all agents strictly prefer bad $b_1$ to $b_2$ and $b_2$ to $b_3$. In the first step, the agents start eating $b_1$ simultaneously, eat half of $b_1$ each, then move onto $b_2$, and eat half of $b_2$ each. Consider an integral allocation in which $b_1$ is allocated to agent $1$ and $b_2$ is allocated to agent $2$. In the second step, each agent eats half of $b_3$. At this point, our algorithm will, with probability $1/2$, allocate $b_3$ to agent $2$. This will result in an integral allocation $A$ in the support where $A_1 = \set{b_1}$, and $A_2 = \set{b_2,b_3}$. It is easy to check that agent $2$ continues to envy agent $1$ even if we remove any one bad from agent $2$'s bundle. Removing two bads from agent 2's bundle removes the envy, however. Hence, this allocation is EF2, but not EF1. In fact, we show that both RPS and its polynomial-time variant are always ex-post EF2 for bads.

\begin{proposition}
	\label{prop:ef2-bads}
	For additive valuation functions over bads, recursive probabilistic serial (RPS) and its polynomial-time variant from \Cref{subsec:RPS_Polytime} produce randomized allocations that are ex-ante envy-free (EF) and ex-post envy-free up to two bads (EF2). 
\end{proposition}
\begin{proof}[Proof Sketch]
The reason why the randomized allocation produced by both algorithms is ex-ante EF is the same as in the case of goods. In each step $t$, the fractional allocation $X^t$ produced by the $\Eating$ procedure is envy-free because agents are forced to eat bads at the same rate and an agent is always eating a bad that she prefers the most among all bads not fully consumed. Since each $X^t$ is envy-free, a probability distribution over all $X^t$-s in different branches is also envy-free. Hence, the expected fractional allocation induced in step $t$ is envy-free for each $t$. Hence, the overall fractional allocation, which is the sum of fractional allocations induced in all rounds, is also envy-free.

For ex-post EF2, let $A$ be an integral allocation produced by RPS (or its polynomial-time variant). We notice that \Cref{lem:sd-ef1-lemma} still holds for the case of bads. That is, for $t < \ceil{m/n}$, if $b_{i,t}$ denotes the bad allocated to agent $i$ in iteration $t$, then agent $i$ (weakly) prefers $b_{i,t}$ to all bads in $M^{t+1}$ (bads unallocated after $t$ iterations of the algorithm). Now, when $t < \floor{m/n}$, we know that another agent $h$ must receive a bad in round $t+1 \le \floor{m/n}$. Hence, for all agents $i,h \in N$ and all $t < \floor{m/n}$, we have $v_i(b_{i,t}) \ge v_i(b_{h,t+1})$. Summing over $t < \floor{m/n}$, we get $v_i(\set{b_{i,t} : 1 \le t < \floor{m/n}}) \ge v_i(\set{b_{h,t} : 2 \le t \le \floor{m/n}}) \ge v_i(A_h)$. Note that in the LHS, there are at most two bads missing from $A_i$: the bad allocated to agent $i$ in round $\floor{m/n}$, and a bad which may be allocated to agent $i$ in the final round $\ceil{m/n}$. Hence, we have that $v_i(A_i \setminus S_i) \ge v_i(A_h)$ for some $S_i \subseteq A_i$ with $|S_i| \le 2$. 
\end{proof}

Once again, note that because these algorithms only take the ordinal preferences as input, they in fact guarantee ex-ante SD-EF + ex-post SD-EF2 (which can be defined similarly to SD-EF1 using the SD preference relation). 

It is natural to ask whether this is the best we can do. Note that in the proof of \Cref{prop:ef2-bads}, an ``up to two bads'' relaxation was required because the envying agent may receive an extra bad compared to the envied agent ($\ceil{m/n}$ bads as opposed to $\floor{m/n}$ bads). However, when $m$ is a multiple of $n$, it turns out that the allocation is in fact ex-post EF1, as the following lemma shows. 

\begin{lemma}\label{lem:ef1-bads}
	For additive valuation functions over bads, recursive probabilistic serial and its polynomial-time variant from \Cref{subsec:RPS_Polytime} produce randomized allocations that are ex-ante envy-free (EF) and ex-post envy-free up to one bad (EF1) whenever the number of bads is divisible by the number of agents.
\end{lemma}
\begin{proof}
	Envy-freeness follows directly from Proposition~\ref{prop:ef2-bads}. For EF1, note that the exact same reasoning as in the proof of Proposition~\ref{prop:ef2-bads} applies, and we have $v_i(\set{b_{i,t} : 1 \le t < \floor{m/n}}) \ge v_i(\set{b_{h,t} : 2 \le t \le \floor{m/n}}) \ge v_i(A_h)$ for every pair of agents $i,h \in N$. However, since $n$ divides $m$, we are guaranteed that there is only one bad missing from the LHS, which is that allocated to agent $i$ in round $\floor{m/n} = m/n$. Hence, we have $v_i(A_i \setminus {b_{i, m/n}}) \ge v_i(A_h)$.
\end{proof}

What about the case where $m$ is not a multiple of $n$? It turns out that a simple modification of RPS (or its polynomial-time variant) works. One can simply add a number of dummy bads (which all agents value at $0$) such that the total number of bads becomes divisible by the number of agents $n$, then use RPS or its polynomial-time variant to compute an ex-ante EF + ex-post EF1 allocation (\Cref{lem:ef1-bads}), and finally remove the dummy bads from the computed allocation, which does not affect its fairness guarantees. This yields the following result. 

\begin{theorem}\label{thm:ef1-bads}
	For additive valuation functions over bads, a randomized allocation that is ex-ante envy-free (EF) and ex-post envy-free up to one bad (EF1) always exists, and can be computed in strongly polynomial time. 
\end{theorem}

\paragraph{The case of mixed manna.} Let us briefly examine the \emph{mixed manna} setting in which we need to allocate \emph{items}, and each item can be either a good or a bad for each agent. We generically denote an item by $o$. That is, for each agent $i \in N$ and item $o \in M$, we only require that $v_{i,j} \in \mathbb{R}$. 

Envy-freeness up to one item can be generalized to the mixed manna setting. It requires that if agent $i$ envies agent $h$, then the envy can be eliminated by either removing one item from agent $h$'s bundle (which must be a good for $i$), or one item from agent $i$'s bundle (which must be a bad for $i$).

\begin{definition}[Envy-Freeness Up To One Item (EF1)~\cite{aziz2019fair}]
	An integral allocation $A$ is \emph{envy-free up to one item} if for every pair of agents $i,h \in N$, either $i$ does not envy $h$ (i.e., $v_i(A_i) \geq v_i(A_h)$), or there exists an item $o \in A_i \cup A_h$ such that $v_i(A_i \setminus \{ o \}) \geq v_i(A_h \setminus \{ o \})$.
\end{definition}

It is easy to check that when the instance consists of only goods or only bads, this definition reduces to the EF1 definitions we previously introduced. However, our algorithms achieve a relaxation of this, which we call \emph{weak envy-freeness up to one item} (w-EF1), that allows \emph{both} the removal of an item from agent $h$'s bundle and an item from agent $i$'s bundle. Note that when the instance consists of only goods or only bads, w-EF1 still reduces to the standard EF1 definition.

\begin{definition}[Weak Envy-Freeness Up to One Item (w-EF1)]
	An integral allocation $A$ is \emph{weak envy-free up to one item} if for every pair of agents $i,h \in N$, either $i$ does not envy $h$ (i.e., $v_i(A_i) \geq v_i(A_h)$), or there exists a good $o_h \in A_h$ and/or a bad $o_i \in A_i$ such that $v_i(A_i \setminus \set{o_i}) \geq v_i(A_h \setminus \set{o_h})$.
\end{definition}

While we were unable to settle the existence of an ex-ante EF + ex-post EF1 allocation for the mixed manna setting, we are able to guarantee ex-ante EF + ex-post w-EF1.

\begin{theorem}
	For additive valuation functions in the mixed manna setting, a randomized allocation that is ex-ante envy-free (EF) and ex-post weak envy-free up to one item (w-EF1) always exists, and can be computed in strongly polynomial time.
\end{theorem}
\begin{proof}[Proof Sketch]
Once again, we add dummy items (which all agents value at $0$) to guarantee that $m$ is an integer multiple of $n$, and remove them at the end. When we run RPS (or its polynomial-time variant), we can again apply the same argument as in the proof of Proposition~\ref{prop:ef2-bads} to obtain $v_i(\set{o_{i,t} : 1 \le t < m/n}) \ge v_i(\set{o_{h,t} : 2 \le t \le m/n})$. That is, $v_i(A_i \setminus \set{o_{i, m/n}}) \ge v_i (A_h \setminus \set{o_{h,1}})$, which is the desired ex-post w-EF1 guarantee. Ex-ante EF follows with the same reasoning as before.
\end{proof}

Since our algorithms only take the ordinal preferences of agents over items as input, this in fact guarantees ex-ante SD-EF and ex-post SD-w-EF1 (which can be defined similarly to SD-EF1 using the SD preference relation). 

Finally, we note that the techniques from this section can also be applied to the decomposition of probabilistic serial (PS) achieved in the subsequent work of \citet{Aziz20} to obtain an ex-ante SD-EF and ex-post SD-w-EF1 allocation of mixed manna that retains additional axiomatic properties of PS. 

\subsection{Ex-ante GF + Ex-post Prop1 + Ex-post EF$_1^1$ for Bads}\label{sec:bads-gf-prop1-ef11}

We now show that our main result from Section~\ref{sec:Prop1_EF11} for the case of goods carries over easily to the case of bads.
This is somewhat striking because it has been observed that many results for the case of goods do not easily carry over to the case of bads~\cite{BMSY17}. For example, while it is known that a competitive equilibrium with equal incomes (CEEI) fractional allocation exists for the case of bads, and such an allocation is envy-free and Pareto optimal, it is no longer obtained by maximizing or minimizing the product of (dis)utilities, and whether such an allocation can be computed in polynomial time is a major open question~\cite{BS19,GM20computing}. However, if such an allocation is given, \citet{BS19} show that one can round it to obtain an integral Prop1+EF$_1^1$+fPO allocation in strongly polynomial time. 

Interestingly, we do not need to re-prove the utility guarantee (\Cref{lem:utility}) for the case of bads. Instead, we can simply reduce the case of bads to the case of goods by considering a modified valuation function $\hat{v}_i = -v_i$ (thus guaranteeing $\hat{v}_i \ge 0$), and since \Cref{lem:utility} bounds the difference between the utility of an agent in an integral allocation and her utility in the original fractional allocation from both below and above, this guarantee remains useful for the case of bads. We then show that this guarantee is sufficient to establish Prop1 and EF$_1^1$ in the case of bads as well. 

\begin{theorem}\label{thm:bads-prop1-ef11}
	For additive valuation functions over bads, there is a strongly polynomial-time algorithm that, given any fractional proportional (Prop) allocation as input, computes an implementation of it using integral allocations that are proportional up to one bad (Prop1). If, in addition, the input is a fractional CEEI allocation, then the integral allocations in the support also satisfy envy-freeness up to one bad more-and-less (EF$_1^1$). 
\end{theorem}
\begin{proof}
	Given an instance $I = (N,M,(v_i)_{i \in N})$ of bad division, let us define the corresponding instance of good division as $\hat{I} = (N,M,(\hat{v}_i)_{i \in N})$ of good division, where $\hat{v}_i = -v_i$. Given a fractional allocation $X$ of $I$, we treat it as an allocation of instance $\hat{I}$, and apply \Cref{lem:utility} to find an implementation of $X$. Let $A^1,\ldots,A^{\ell}$ be integral allocations in its support. Note that the `reduction' to goods is used only for the decomposition step, and not for constructing the relevant fractional allocation.
	
	First, suppose $X$ is Prop under $I$. We want to show that for each $k \in [\ell]$, $A^k$ is Prop1 under $I$. Since $X$ is proportional, for every $i \in N$, $\hat{v}_i(X_i) \le \hat{v}_i(\vec{1}^m)/n$, where $\hat{v}_i(\vec{1}^m)$ is agent $i$'s valuation for receiving all ``goods'' fully. Fix $k \in [\ell]$. By \Cref{lem:utility}, we have that for every agent $i \in N$, either $\hat{v}_i(A^k_i) \le \hat{v}_i(X_i) \le \hat{v}_i(\vec{1}^m)/n$, or there exists a bad $j \in A^k_i$ such that $\hat{v}_i(A^k_i) - \hat{v}_{i,j} < \hat{v}_i(X_i) \le \hat{v}_i(\vec{1}^m)/n$. Therefore, $A^k$ is Prop1 under $I$. Note that here, we are using the second part of \Cref{lem:utility}, whereas in the goods case, we used the first part to establish Prop1.
	
	Next, suppose that $X$ is a fractional CEEI allocation under $I$. Since this allocation is envy-free~\cite{BMSY17}, and therefore proportional, the aforementioned argument still applies for ex-post Prop1 under $I$. We show that in this case, $A^k$ is also EF$_1^1$ under $I$ for each $k \in [\ell]$. As mentioned before, since $X$ is a fractional CEEI allocation under $I$, \Cref{eqn:ceei-bads-condition} holds. Fix a pair of distinct agents $i,h \in N$. By \Cref{lem:utility}, either $\hat{v}_i(A^k_i) \ge \hat{v}_i(X_i)$, or there exists $b_i^- \notin A^k_i$ with $X_{i,b_i^{-}} > 0$ such that $\hat{v}_i(A^k_i \cup \{b_i^-\}) > \hat{v}_i(X_i)$. Similarly, either $\hat{v}_h(A^k_h) \le \hat{v}_h(X_h)$, or there exists $b_h^+ \in A^k_h$ with $X_{h,b_h^{-}} < 1$ such that $\hat{v}_h(A^k_h \setminus \{b_h^+\}) < \hat{v}_h(X_h)$. To simplify the analysis like in the case of goods, let us assume that the second condition holds in both cases.	
	
	By summing the right-hand side inequality in \Cref{eqn:ceei-bads-condition} over all $b \in A^k_i \cup \set{b_i^-}$, we get 
	$$
	\frac{\hat{v}_h(A^k_i \cup \set{b_i^-})}{\hat{v}_h(X_h)} \ge \frac{\hat{v}_i(A^k_i \cup \set{b_i^-})}{\hat{v}_i(X_i)} > 1.
	$$
	Thus, $\hat{v}_h(A^k_i \cup \{b_i^-\}) > \hat{v}_h(X_h) > \hat{v}_h(A^k_h \setminus \set{b_h^+})$, implying that $v_h(A^k_i \cup \{b_i^-\}) < v_h(X_h) < v_h(A^k_h \setminus \set{b_h^+})$, and $A^k$ satisfies EF$_1^1$, as desired.
\end{proof}

\begin{corollary}\label{bads-fgf-prop1}
	For additive valuation functions over bads, there exists a randomized allocation that is ex-ante group fair (GF), ex-post proportional up to one bad (Prop1), and ex-post envy-free up to one bad more-and-less (EF$_1^1$). Further, it can be computed in strongly polynomial time given a fractional CEEI allocation.  
\end{corollary}

We note that in \Cref{bads-fgf-prop1}, unlike in the case of goods, we have to add the condition ``given a fractional CEEI allocation''. This is because a fractional CEEI allocation is known to be computable in strongly polynomial time in the case of goods, but this remains an interesting open question in the case of bads~\cite{BS19,GM20computing}. 

\section{Discussion}\label{sec:Discussion}

Perhaps the most fascinating open question that stems from our work is whether ex-ante envy-freeness (or even ex-ante proportionality) is compatible with ex-post EF1 and ex-post PO.\footnote{While our simulations did not find a counterexample, simulations have been known to be misleading for existential results in the past, e.g., in case of the (non)existence of MMS allocations~\cite{KPW18}.} 

\medskip\noindent\textbf{Open Question:} \emph{Does there always exist a randomized allocation that is ex-ante EF, ex-post EF1, and ex-post PO? What about ex-ante Prop, ex-post EF1, and ex-post PO?}\medskip

The difficulty in approaching this question is that there are very few available methods of finding integral EF1+PO allocations~\cite{CKMP+19,BKV18}, so finding many such allocations and randomizing over them is tricky. Also, unlike the set of integral EF1 allocations, which we somewhat understand, not much is known about the set of integral EF1+PO allocations other than the fact that it is always non-empty. For example, round-robin method allows us to control exactly which agents might envy which other agents in addition to obtaining EF1. The following questions, which thus resolve positively for EF1, are open for EF1+PO to the best of our knowledge. 

\medskip\noindent\textbf{Open Question:} \emph{Given agents $i$ and $j$, does there always exist an integral EF1+PO allocation in which agent $i$ does not envy agent $j$? What about an integral EF1+PO allocation in which agent $i$ does not envy any other agent? Given a priority ordering over the agents, does there always exist an EF1+PO allocation in which no agent envies an agent with lower priority?}\medskip

Various other open problems remain. For instance, other related notions such as envy-freeness up to any good (EFX) and approximate maximin-share guarantee (MMS) are natural candidates for ex-post fairness.\footnote{For additive valuations, whether an integral EFX allocation always exists is an open problem, and while integral MMS allocations do not always exist~\cite{KPW18}, $3/4$-approximate MMS allocations are known to always exist~\cite{GHSS+18,GT20}.} More broadly, the next step would be to achieve ex-ante and ex-post fairness guarantees simultaneously in a variety of other problems such as voting, matching, and public decision-making.




\appendix
\section*{Appendix}
\section{Decomposition result of BUDISH ET AL.~\citep{BCKM13}}
\label{sec:BCKMRounding}

Let $X$ be a fractional allocation. Recall that $X$ satisfies \emph{column-wise} feasibility constraints, namely $0 \leq \sum_{i \in N} X_{i,j} \leq 1$ for all $j \in M$. More generally, we can impose capacity constraints of the form $\underline{q}_S \leq \sum_{(i,j) \in S} X_{i,j} \leq \overline{q}_S$, where $S$ is a \emph{constraint set} comprising of a collection of agent-object pairs, and $\underline{q}_S$ and $\overline{q}_S$ are the \emph{lower} and \emph{upper quotas} for $S$, respectively. The set of all capacity constraints imposed by a given problem is called the \emph{constraint structure} $\H$ of the problem, and is specified as a collection of all constraint sets and the corresponding quotas $(\underline{q}_S,\overline{q}_S)_{S \in \H}$. Given a constraint structure $\H$, we say that the fractional allocation $X$ admits a \emph{feasible} implementation $\mathbf{X} \coloneqq \{(p^k,A^k)\}_{k \in [\ell]}$ if every integral allocation in its support also satisfies the constraints in $\H$. That is, for every $k \in [\ell]$, we have 
$$\textstyle{ \underline{q}_S \leq \sum_{(i,j) \in S} A^k_{i,j} \leq \overline{q}_S \text{ for every } S \in \H. }$$

\begin{definition}[Hierarchy and bihierarchy]
	A constraint structure $\H$ is said to be a \emph{hierarchy} (or a laminar family) if for every $S,S' \in \H$, we have that either $S \subset S'$, or $S' \subset S$, or $S \cap S' = \emptyset$. We say that $\H$ is a \emph{bihierarchy} if it can be partitioned into two hierarchies, i.e., if there exist hierarchies $\H_1$ and $\H_2$ such that $\H = \H_1 \cup \H_2$ and $\H_1 \cap \H_2 = \emptyset$.
	\label{defn:bihierarchy}
\end{definition}
As an example, consider the fractional allocation $X$ in \Cref{fig:prop_decomposition}. The row constraints (shown as red or blue solid rectangles) as well as all singleton constraints of the form $0 \leq X_{i,j} \leq 1$ (not shown in the figure) together constitute a hierarchy, say $\H_1$, since for any pair of constraint sets, either they are disjoint or one is completely contained inside the other. Similarly, the column constraints (shown as gray dotted rectangles) form another hierarchy $\H_2$. Furthermore, $\H \coloneqq \H_1 \cup \H_2$ is a bihierarchy since any constraint set (rectangle or singleton) belongs to exactly one of $\H_1$ or $\H_2$.

\begin{figure}[ht]
\tikzset{every picture/.style={line width=0.3pt}}
\tikzset{%
  RedRectangle/.style={rectangle,rounded corners,draw=red,line width=0.3pt}
}
\tikzset{%
  BlueRectangle/.style={rectangle,rounded corners,draw=blue,line width=0.3pt}
}
\tikzset{%
  LightRedRectangle/.style={rectangle,rounded corners,draw=red,line width=0.1pt}
}
\tikzset{%
  LightBlueRectangle/.style={rectangle,rounded corners,draw=blue,line width=0.1pt}
}
\tikzset{%
  GrayRectangle/.style={rectangle,rounded corners,draw=gray,line width=0.5pt,densely dotted}
}
\begin{tikzpicture}
\footnotesize
\def\x{-3};
        \matrix [matrix of math nodes,left delimiter=(,right delimiter=)] (m) at (\x+3.5,0) 
        {
            0.6 & { } & 0.4 & { } & 0.4 & { } & 0.6 \\  
            { }& { } & { }& { } & { }& { } & { }\\
            { }& { } & { }& { } & { }& { } & { }\\
            0.4 & { } & 0.6 & { } & 0.6 & { } & 0.4 \\ 
        };
        \node[above=0.2cm of m] {$X$};
        \node[RedRectangle,inner sep=1pt,fit=(m-1-1.north west)(m-1-1.south east)] {};
        \node[RedRectangle,inner sep=2pt,fit=(m-1-1.north west)(m-1-3.south east)] {};
        \node[RedRectangle,inner sep=3pt,fit=(m-1-1.north west)(m-1-5.south east)] {};
        \node[RedRectangle,inner sep=4pt,fit=(m-1-1.north west)(m-1-7.south east)] {};
        \node[BlueRectangle,inner sep=1pt,fit=(m-4-3.north west)(m-4-3.south east)] {};
        \node[BlueRectangle,inner sep=2pt,fit=(m-4-3.north west)(m-4-5.south east)] {};
        \node[BlueRectangle,inner sep=3pt,fit=(m-4-3.north west)(m-4-7.south east)] {};
        \node[BlueRectangle,inner sep=4pt,fit=(m-4-1.north west)(m-4-7.south east)] {};
%
		\node[GrayRectangle,inner ysep=8pt,inner xsep=0pt,fit=(m-1-1.north west)(m-4-1.south east)] {};
        \node[GrayRectangle,inner ysep=8pt,inner xsep=0pt,fit=(m-1-3.north west)(m-4-3.south east)] {};
        \node[GrayRectangle,inner ysep=8pt,inner xsep=0pt,fit=(m-1-5.north west)(m-4-5.south east)] {};
        \node[GrayRectangle,inner ysep=8pt,inner xsep=0pt,fit=(m-1-7.north west)(m-4-7.south east)] {};
%
		\node at (\x+5.4,0) {=};
%
		\node at (\x+5.75,0) {0.4};
%
        \matrix [matrix of math nodes,left delimiter={[},right delimiter={]}] (m) at (\x+7.2,0) 
        {
            1 & { } & 0 & { } & 1 & { } & 0 \\  
            { }& { } & { }& { } & { }& { } & { }\\
            { }& { } & { }& { } & { }& { } & { }\\
            0 & { } & 1 & { } & 0 & { } & 1 \\ 
        };
        \node[above=0.2cm of m] {$A^1$};
        \node[LightRedRectangle,inner sep=1pt,fit=(m-1-1.north west)(m-1-1.south east)] {};
        \node[LightRedRectangle,inner sep=2pt,fit=(m-1-1.north west)(m-1-3.south east)] {};
        \node[LightRedRectangle,inner sep=3pt,fit=(m-1-1.north west)(m-1-5.south east)] {};
        \node[LightRedRectangle,inner sep=4pt,fit=(m-1-1.north west)(m-1-7.south east)] {};
        \node[LightBlueRectangle,inner sep=1pt,fit=(m-4-3.north west)(m-4-3.south east)] {};
        \node[LightBlueRectangle,inner sep=2pt,fit=(m-4-3.north west)(m-4-5.south east)] {};
        \node[LightBlueRectangle,inner sep=3pt,fit=(m-4-3.north west)(m-4-7.south east)] {};
        \node[LightBlueRectangle,inner sep=4pt,fit=(m-4-1.north west)(m-4-7.south east)] {};
%
		\node at (\x+8.5,0) {+};
%
		\node at (\x+8.85,0) {0.2};
%
        \matrix [matrix of math nodes,left delimiter={[},right delimiter={]}] (m) at (\x+10.3,0) 
        {
            1 & { } & 0 & { } & 0 & { } & 1 \\  
            { }& { } & { }& { } & { }& { } & { }\\
            { }& { } & { }& { } & { }& { } & { }\\
            0 & { } & 1 & { } & 1 & { } & 0 \\ 
        };
        \node[above=0.2cm of m] {$A^2$};
        \node[LightRedRectangle,inner sep=1pt,fit=(m-1-1.north west)(m-1-1.south east)] {};
        \node[LightRedRectangle,inner sep=2pt,fit=(m-1-1.north west)(m-1-3.south east)] {};
        \node[LightRedRectangle,inner sep=3pt,fit=(m-1-1.north west)(m-1-5.south east)] {};
        \node[LightRedRectangle,inner sep=4pt,fit=(m-1-1.north west)(m-1-7.south east)] {};
        \node[LightBlueRectangle,inner sep=1pt,fit=(m-4-3.north west)(m-4-3.south east)] {};
        \node[LightBlueRectangle,inner sep=2pt,fit=(m-4-3.north west)(m-4-5.south east)] {};
        \node[LightBlueRectangle,inner sep=3pt,fit=(m-4-3.north west)(m-4-7.south east)] {};
        \node[LightBlueRectangle,inner sep=4pt,fit=(m-4-1.north west)(m-4-7.south east)] {};
%
		\node at (\x+11.6,0) {+};
%
		\node at (\x+11.95,0) {0.4};
%
        \matrix [matrix of math nodes,left delimiter={[},right delimiter={]}] (m) at (\x+13.4,0) 
        {
            0 & { } & 1 & { } & 0 & { } & 1 \\  
            { }& { } & { }& { } & { }& { } & { }\\
            { }& { } & { }& { } & { }& { } & { }\\
            1 & { } & 0 & { } & 1 & { } & 0 \\ 
        };
        \node[above=0.2cm of m] {$A^3$};
        \node[LightRedRectangle,inner sep=1pt,fit=(m-1-1.north west)(m-1-1.south east)] {};
        \node[LightRedRectangle,inner sep=2pt,fit=(m-1-1.north west)(m-1-3.south east)] {};
        \node[LightRedRectangle,inner sep=3pt,fit=(m-1-1.north west)(m-1-5.south east)] {};
        \node[LightRedRectangle,inner sep=4pt,fit=(m-1-1.north west)(m-1-7.south east)] {};
        \node[LightBlueRectangle,inner sep=1pt,fit=(m-4-3.north west)(m-4-3.south east)] {};
        \node[LightBlueRectangle,inner sep=2pt,fit=(m-4-3.north west)(m-4-5.south east)] {};
        \node[LightBlueRectangle,inner sep=3pt,fit=(m-4-3.north west)(m-4-7.south east)] {};
        \node[LightBlueRectangle,inner sep=4pt,fit=(m-4-1.north west)(m-4-7.south east)] {};
    \end{tikzpicture}
    \caption{Decomposition of a fractional proportional allocation $X$ into integral Prop1 allocations $A^1$, $A^2$, and $A^3$. The underlying fair division instance comprises of four goods and two agents with valuations $v_{1,1}=10$, $v_{1,2}=6$, $v_{1,3}=4$, $v_{1,4}=2$ and $v_{2,1}=2$, $v_{2,2}=10$, $v_{2,3}=6$, $v_{2,4}=4$.}
    \label{fig:prop_decomposition}
\end{figure}
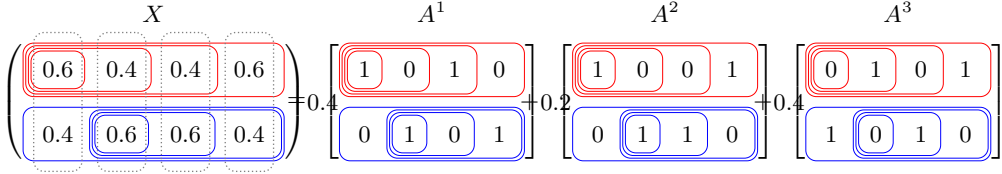

\citet{BCKM13} showed that bihierarchy constraint structure is a sufficient condition for a fractional allocation to admit a feasible implementation. We recall this result in \Cref{prop:BCKMRounding}.

\begin{proposition}[\citet{BCKM13}]\label{prop:BCKMRounding}
	Given a fractional allocation $X$ satisfying a bihierarchy constraint structure $\H$, one can compute, in strongly polynomial time, a set of coefficients $p_1,\ldots,p_{\ell} \in [0,1]$ and integral allocations $A^1,\ldots,A^{\ell}$ such that (a) $\sum_{k=1}^{\ell} p^k = 1$, (b) each $A^k$ satisfies the constraints in $\H$, and (c) $X = \sum_{k=1}^{\ell} p^k A^k$.
\end{proposition}

Observe that the well-known Birkhoff-von Neumann theorem is a special case of \Cref{prop:BCKMRounding} when $\H_1$ consists of all singleton as well as row constraints, $\H_2$ consists of all column constraints, and the lower and upper quotas for each constraint set are $0$ and $1$, respectively. It is worth pointing out that while \citet{BCKM13} only note a polynomial running time, it is easy to check that their (combinatorial) algorithm, in fact, runs in \emph{strongly} polynomial time.

\citet{BCKM13} use \Cref{prop:BCKMRounding} to establish the utility guarantee (\Cref{prop:BCKMutility}).

\section{Omitted Material from Section~\ref{sec:Prop1_EF11}}

For completeness, we provide a proof of the result mentioned by~\citet{CFSV19}, that MNW satisfies group fairness. Indeed, we prove a stronger result. 
We show that this holds not only for the case of goods division, but in the more general cake-cutting setting.

In the cake-cutting setting, we have a single heterogeneous good modeled by $M=[0,1]$, and each agent has an integrable valuation function $v_i: [0,1] \to \mathbb{R}_+$. Each agent's allocation, $A_i$, is a finite union of intervals, and $v_i(A_i)$ is given by the area under the curve $v_i$ over the intervals in $A_i$.

A price function $p$ is a measure on $M$. Say that a subset $Z \subseteq M$ is a \emph{positive slice} if there exists an agent $i$ with $v_i(Z)>0$. A pair $(X,p)$ of allocation $X$ and price $p$ is a \emph{strong competitive equilibrium from equal incomes} (s-CEEI)~\citep{segal2019monotonicity} if the following conditions hold:
\begin{enumerate}
	\item For every positive slice $Z \subseteq M$, $p(Z)>0$.
	\item For every pair of positive slices $Z \subseteq M$ and $Z_i \subseteq X_i$, $\frac{v_i(Z_i)}{p(Z_i)} \ge \frac{v_i(Z)}{p(Z)}$.
	\item For every $i \in N$, $p(X_i)=1$
\end{enumerate}
\citet{segal2019monotonicity} show that the MNW allocation is an s-CEEI allocation when $p$ is the standard price measure defined by, for all $i \in N$ and all $Z_i \subseteq X_i$, $p(Z_i) = \frac{v_i(Z_i)}{v_i(X_i)}$.

\begin{proposition}\label{prop:mnw-properties}
	MNW satisfies group fairness in the cake-cutting setting.
\end{proposition}
\begin{proof}
	Consider an instance of cake cutting and suppose that the MNW/s-CEEI allocation $X$ violates group fairness. Then there exist subsets of agents $S,T \subseteq N$ with $|T| \neq 0$ and an allocation $Y$ of $\cup_j X_j$ among agents in $S$ such that $\frac{|S|}{|T|} \cdot v_i(Y_i) \ge v_i(X_i)$ for all agents $i \in S$ and the inequality is strict for at least one $i \in S$. 
	
	By the s-CEEI condition (2), we know that $v_i(X_i)=\frac{v_i(X_i)}{p(X_i)} \ge \frac{v_i(Y_i)}{p(Y_i)}$ for all $i \in S$. Combining this with the group fairness condition yields
	\[ \frac{|S|}{|T|} \cdot v_i(Y_i) \ge \frac{v_i(Y_i)}{p(Y_i)} \implies p(Y_i) \ge \frac{|T|}{|S|} \]
	for all $i \in S$, with at least one inequality strict. Taking the sum over all $i \in S$ gives us $\sum_{i \in S} p(Y_i) > |T|$. Finally, because $\sum_{i \in S} p(Y_i) = p(\cup_{i \in S} Y_i) = p(\cup_{j \in T} X_j)=\sum_{j \in T} p(X_j)$, it must be the case that $p(X_j)>1$ for some $j \in T$, violating the assumption that $(X,p)$ is s-CEEI.
	\end{proof}
	
\section{Omitted Material from Section~\ref{sec:characterization}}

\subsection{Maximum Nash Welfare Variant}

\paragraph{Relaxing group fairness.} 
We now present a replication-imnvariant allocation rule that is not a refinement of MNW. As a consequence of Theorem~\ref{thm:mnw-characterization}, it therefore fails group fairness. However, we show that it satifies a relaxation of group fairness that still generalizes many of the fairness properties common in the literature.

Say that good $g$ is a \emph{weak} good if there exists only a single agent with $v_i(g)>0$. Otherwise, we say that $g$ is a \emph{strong} good. Let $W_i = \{ g \in M : v_j(g)=0 \; \forall j \neq i \}$ denote the set of all weak goods valued positively by agent $i$. 

Mechanism \emph{Maximum Nash Welfare Variant} (MNW-V) works as follows. Given an instance $I$, let $I_{-W}$ denote the instance with weak goods removed. Then MNW-V$(I)_i = MNW(I_{-W})_i \cup W_i$. That is, we begin by allocating only the strong goods according to MNW. We then allocate to each agent the weak goods that she values positively.

\begin{example}
	Suppose there are three agents and two goods, $g$ and $g^*$, with $v_1(g)=v_2(g)=v_3(g)=1$ and $v_1(g^*)=1$ and $v_2(g^*)=v_3(g^*)=0$. Note that $g$ is a strong good and $g^*$ is a weak good, valued only by agent 1. MNW allocates agents 1 and 2 half of $g$ each, and all of $g^*$ to agent 3. MNW-V, on the other hand, allocates each agent a $1/3$ share of good $g$, and all of $g^*$ to agent 3. Note that MNWV does not satisfy group fairness. If we let $S=N$ and $T$ consist only of agent 3, then we can partition agent 3's bundle so that $B_1 = B_2= \frac{1}{6}g$, and $B_3 = g^*$, for utilities $v_1(B_1)=v_2(B_2)=1/6$ and $v_3(B_3)=1$. When these utilities are scaled by a factor of $|S|/|T|=3$, each agent is strictly happier than under the original MNW-V allocation.
\end{example}

The above example shows that MNW-V violates group fairness. But we can show that the allocation is fair to any group $S$ provided that $S$ is small relative to group $T$. We formalize this using the following property.

\begin{definition}
	An allocation $X$ is \emph{group fair for less} if for all non-empty subsets of agents $S,T \subseteq N$ with $|S| \le |T|$, there is no allocation $Y$ of the set of goods $\cup_{i \in T} X_i$ among agents in $S$ such that $\frac{|S|}{|T|} \cdot v_i(Y_i) \ge v_i(X_i)$ for all agents $i \in S$ and at least one inequality is strict.
\end{definition}

In particular, group fairness for less generalizes proportionality, envy-freeness, Pareto optimality, group envy-freeness~\cite{Var74,BTD92}, and the core~\cite{Scarf62}.

We show that MNW-V satisfies group fairness for less. As a consequence, a characterization of MNW in the additive utility setting therefore seems to require the full power of group fairness.

\begin{theorem}
MNW-V satisfies group fairness for less.
\end{theorem}

\begin{proof}
Denote the MNW-V allocation on some instance $I$ by $A$, and fix sets $S$ and $T$ with $|S| \le |T|$. Since we know that MNW is group fair, for every partition $(B_i)_{i \in S}$ of $\cup_{i \in T} (A_i \backslash W_i)$, we have either
\begin{equation} \label{eq:condition1}
 \frac{|S|}{|T|} \cdot v_i(B_i) < v_i(A_i \backslash W_i)
\end{equation}
for some $i \in S$, or 
\begin{equation} \label{eq:condition2}
 \frac{|S|}{|T|} \cdot v_i(B_i) = v_i(A_i \backslash W_i)
\end{equation}
for all $i \in S$.

 Fix a partition $(C_i)_{i \in S}$ of $\cup_{i \in T} A_i$, and for every $i \in S$ denote by $B_i$ the intersection of $C_i$ with the set of strong goods. $(B_i)_{i \in S}$ is a partition of  $\cup_{i \in T} (A_i \backslash W_i)$, so we must be in one of the two cases above. Suppose we are in the first, with witness agent $i$. If $i \in T$, then $C_i \subseteq B_i \cup W_i$ and we have
\[ 
\frac{|S|}{|T|} \cdot v_i(C_i) \le \frac{|S|}{|T|} \cdot v_i(B_i) + \frac{|S|}{|T|} \cdot v_i(W_i) < v_i(A_i \backslash W_i) + v_i(W_i) = v_i(A_i),
\]
where the strict inequality follows from Equation~\ref{eq:condition1}.
If $i \not \in T$, then $v_i(C_i)=v_i(B_i)$ (since no goods from $W_i$ are contained in $\cup_{i \in T} A_i$), and we have
\[ 
\frac{|S|}{|T|} \cdot v_i(C_i) = \frac{|S|}{|T|} \cdot v_i(B_i) < v_i(A_i \backslash W_i) \le v_i(A_i). 
\] 
In both cases, group fairness holds with respect to sets $S$ and $T$.

A similar argument can be used if we are in the case corresponding to Equation~\ref{eq:condition2}, with the relevant strict inequalities replaced by equalities.
\end{proof}

\subsection{Fractional Group Fairness Characterizes MNW}

We now show that the dependence on replication invariance can be removed from Theorem~\ref{thm:mnw-characterization} if group fairness is strengthened to \emph{fractional group fairness}. Intuitively, fractional group fairness extends group fairness to hold even for groups of agents that may consist of (hypothetical) fractional agents.
Denote a fractional subset of agents $S=(x^S_i)_{i \in N}$, where $x^S_i \in [0,1]$ denotes the fraction of agent $i$ that is included in $S$. Define $|S|=\sum_{i \in N} x^S_i$.

\begin{definition}[Fractional Group Fairness (FGF)]
	An allocation $X$ is \emph{fractionally group fair} if for all fractional subsets of agents $S,T \subseteq N$ with $|T| \neq 0$, there is no allocation $Y$ of the set of goods $\cup_{j} x^T_j X_j$ among all agents such that $\frac{|S|}{|T|} \cdot v_i(Y_i) \ge x^S_i \cdot v_i(X_i)$ for all agents $i \in N$ and the inequality is strict for at least one $i \in N$ with $x^S_i > 0$.
\end{definition}

Note that we recover group fairness (\Cref{defn:GF}) as a special case of FGF by restricting to the subsets $S$ and $T$ with $x^S_i \in \{0,1\}$ and $x^T_i \in \{0,1\}$ for all $i \in N$.

When we strengthen group fairness to fractional group fairness, we obtain a characterization of MNW that does not require replication invariance.

\begin{theorem} \label{thm:fgf-characterizes-mnw}
Any allocation rule that satisfies fractional group fairness must be a refinement of MNW.
\end{theorem}

\begin{proof}
The proof proceeds in a similar manner to the proof of Theorem~\ref{thm:mnw-characterization}. Let $f$ be an allocation rule, and suppose that there exists an instance $I$ for which there exists an allocation $A \in f(I)$ but $A \not \in MNW(I)$. If $A$ is not Pareto optimal then it violates group fairness. So suppose $A$ is Pareto optimal.

By Lemma~\ref{lem:local-implies-global-mnw}, there exist agents $i,j$ and $X \subseteq A_i$ with $v_i(X) \in (0,v_i(A_i))$ and $v_j(X)>0$ such that
\begin{equation}\label{eq:violation2}
 v_j(X) > v_j(A_j) \cdot \frac{v_i(X)}{v_i(A_i-X)}.
\end{equation}

Let $T= \{i \}$ (that is, $x^T_i=1$ and $x^T_k=0$ for all $k \neq i$) and $S=\{ i \} \cup \{ \frac{v_i(X)}{v_i(A_i-X)} j \}$ (that is, $x^S_i=1$, $x_j^S= \frac{v_i(X)}{v_i(A_i-X)}$, and $x^S_k=0$ otherwise). Let $B_i=A_i-X$ and $B_j=X$. Then 
\begin{align*}
\frac{|S|}{|T|} \cdot v_i(B_i) = \left(1+\frac{v_i(X)}{v_i(A_i-X)}\right) \cdot v_i(A_i-X)
= v_i(A_i-X)+v_i(X)=x^T_i \cdot v_i(A_i)
\end{align*}
and 
\begin{align*}
\frac{|S|}{|T|} \cdot v_j(B_j) \cdot \frac{1}{x^S_j} &= \left(1+\frac{v_i(X)}{v_i(A_i-X)}\right) \cdot  \frac{v_i(A_i-X)}{v_i(X)} \cdot v_j(X)\\
&= \frac{v_i(A_i)}{v_i(X)} \cdot v_j(X) > \frac{v_i(A_i-X)}{v_i(X)} \cdot  v_j(X)>v_j(A_j),
\end{align*}
where the final inequality follows from Inequality~\ref{eq:violation2}.
Fractional group fairness is violated.
\end{proof}

As a consequence of Theorem~\ref{thm:fgf-characterizes-mnw} and Proposition~\ref{prop:mnw-properties}, we know that the set of Maximum Nash Welfare allocations coincides exactly with the set of fractionally group fair allocations. We note that such a statement does not hold if we use only the more permissive axiom of group fairness.

\begin{example}
Consider an allocation problem with two goods and two agents. We have $v_1(g_1)=v_2(g_2)=1$, and $v_1(g_2)=v_2(g_1)=\epsilon$ for some small $\epsilon>0$. The only MNW allocation is to assign $g_1$ to agent 1 and $g_2$ to agent 2. But consider allocation $A$ that sets $A_1=g_1 + \frac{1}{3}g_2$ and $A_2=\frac{2}{3}g_2$. It is easy to check that $A$ is group fair despite not being an MNW allocation.
\end{example}

\section{Recursive Probabilistic Serial}
\label{sec:RPS_recursive_formulation}

\begin{algorithm}[!h]
	\caption{\Eating}
	\label{alg:Eating}
	\begin{algorithmic}[1]
		\State \textbf{Input}: A (possibly fractional) set of goods $M'$.
		\State \textbf{Output}: A fractional allocation $X$, integral allocations $\{A^k\}_{k \in [\ell]}$, and coefficients $\{w^k\}_{k \in [\ell]}$
		\State $X \leftarrow$ fractional (partial) allocation of $M'$ by running Probabilistic Serial for one unit of time\footnote{One unit is the time it takes any agent to consume a single unit of any good.}
		\State \textbf{return} $X$
	\end{algorithmic}
\end{algorithm}

\begin{algorithm}[!h]
	\caption{Recursive Probabilistic Serial (\RPS)}
	\label{alg:RPS}
	\begin{algorithmic}[1]
		\State \textbf{Input}: A (possibly fractional) set of goods $M^t$.
		\State \textbf{Output}: A fractional allocation $Y^t$.
		\If{$M^t = \emptyset$}
			\State $Y^t \leftarrow \emptyset$
		\Else
			\State $X^t \leftarrow \Eating(M^t,1)$
			\State BvN: $X^t = \sum_{k=1}^\ell w^{t,k} A^{t,k}$
			\State $Y^t \leftarrow X^t + \sum_{k = 1}^\ell w^{t,k} \cdot \RPS(M^t \setminus A^{t,k})$
		\EndIf
		\State \textbf{return} $Y^t$
	\end{algorithmic}
\end{algorithm}

\end{document}